\documentclass[journal]{IEEEtran}

\usepackage{hyperref,doi,url}
\hypersetup{colorlinks=true, linkcolor=blue, breaklinks=true, urlcolor=blue}

\usepackage[table]{xcolor}
\usepackage{graphicx}
\usepackage{amsmath}
\usepackage{amssymb}
\usepackage{amsthm}
\usepackage{bm}
\usepackage[compress]{cite}
\usepackage{enumitem}
\usepackage{dsfont}
\usepackage{bbm}
\usepackage{bbold}
\usepackage{indentfirst}
\usepackage{subfigure}
\usepackage{hhline}
\usepackage{mathtools}

\usepackage{tikz}
\usepackage{pgfplots}
\usepackage[prependcaption,colorinlistoftodos]{todonotes}

\usepackage{lineno}
\usepackage{multirow}
\usepackage{amsfonts}
\usepackage{textcomp}
\usepackage{stfloats}

\usepackage{tikz}
\usetikzlibrary{arrows.meta}
\tikzset{%
	>={Latex[width=2mm,length=2mm]},
	base/.style = {rectangle, rounded corners, draw=black,
		minimum width=2cm, minimum height=1cm,
		text centered, font=\sffamily},
	activityStarts/.style = {base, fill=blue!30},
	startstop/.style = {base, fill=red!30},
	activityRuns/.style = {base, fill=green!30},
	process/.style = {base, minimum width=2.5cm, fill=orange!15,
		font=\ttfamily},
}

\usepackage{centernot}
\usepackage{version,xspace}
\usepackage{environ}
\usepackage{blkarray}
\usetikzlibrary{automata,positioning,arrows,through}

\usepackage{threeparttable}

\usepackage{float}
\usepackage{indentfirst}
\usepackage{mathrsfs}
\usepackage[linesnumbered,ruled,vlined]{algorithm2e}

\usepackage{algorithmic}
\usepackage{tabularx}

\usepackage{cuted}

\usepackage[T1]{fontenc}

\usepackage[nocomma]{optidef}

\usepackage{tikz}
\usetikzlibrary{automata, positioning}
\usetikzlibrary{shapes,backgrounds}
\usetikzlibrary{positioning, shapes.geometric}

\usetikzlibrary{shapes,arrows}
\usetikzlibrary{fit}
\tikzset{%
	block/.style    = {draw, thick, rectangle, minimum height = 2.3em,
		minimum width = 3em},
	sum/.style      = {draw, circle, node distance = 2cm}, % Adder
	input/.style    = {coordinate}, % Input
	output/.style   = {coordinate} % Output
}

\newtheorem{thom}{\textbf{Theorem}}
\newtheorem{assumption}{\textbf{Assumption}}
\newtheorem{remark}{\textbf{Remark}}

\newtheorem{lema}{\textbf{Lemma}}

\newtheorem{problem}{\textbf{Problem}}

\newtheorem{example}{\textbf{Example}}

\DeclareMathOperator*{\argmax}{arg\,max}
\DeclareMathOperator*{\argmin}{arg\,min}

\DeclareMathOperator*{\minimize}{\rm{minimize}}
\DeclareMathOperator*{\maximize}{\rm{maximize}}
\DeclareMathOperator*{\st}{\rm{subject~to}}
\DeclareMathOperator{\E}{\mathbb{E}}

\begin{document}

\title{Privacy-Utility Trade-Offs Against\\Limited Adversaries}
\author{Xiaoming Duan, Zhe Xu, Rui Yan, Ufuk Topcu
\thanks{This work was
	supported by ONR N00014-21-1-2502, AFRL FA9550-19-1-0169, and DARPA D19AP00004.}
\thanks{Xiaoming Duan is with the Oden Institute for Computational Engineering and Sciences, The University of Texas at Austin, Austin, TX, 78712, USA. email: {\tt\small xiaomingduan.zju@gmail.com}.}
\thanks{Zhe Xu is with the School for Engineering of Matter, Transport, and Energy, Arizona State University, Tempe, AZ, 85287, USA. email: {\tt\small xzhe1@asu.edu}.
}
\thanks{Rui Yan is with the Department of Computer Science, University of Oxford, Oxford OX1 3QD, UK. email: {\tt\small rui.yan@cs.ox.ac.uk}.}
\thanks{Ufuk Topcu is with the Department of Aerospace Engineering and Engineering Mechanics, The University of Texas at Austin, Austin, TX, 78712, USA. email: {\tt\small utopcu@utexas.edu}.}
}

% make the title area
\maketitle

\IEEEpeerreviewmaketitle
\begin{abstract}
We study privacy-utility trade-offs where users share privacy-correlated useful information with a service provider to obtain some utility. The service provider is adversarial in the sense that it can infer the users' private information based on the shared useful information. To minimize the privacy leakage while maintaining a desired level of utility, the users carefully perturb the useful information via a probabilistic privacy mapping before sharing it. We focus on the setting in which the adversary attempting an inference attack on the users' privacy has potentially biased information about the statistical correlation between the private and useful variables. This information asymmetry between the users and the limited adversary leads to better privacy guarantees than the case of the omniscient adversary under the same utility requirement. We first identify assumptions on the adversary's information so that the inference costs are well-defined and finite. Then, we characterize the impact of the information asymmetry and show that it increases the inference costs for the adversary. We further formulate the design of the privacy mapping against a limited adversary using a difference of convex functions program and solve it via the concave-convex procedure. When the adversary's information is not precisely available, we adopt a Bayesian view and represent the adversary's information by a probability distribution. In this case, the expected cost for the adversary does not admit a closed-form expression, and we establish and maximize a lower bound of the expected cost. We provide a numerical example regarding a census data set to illustrate the theoretical results.
\end{abstract}

\section{Introduction}
\paragraph*{Problem description and motivation} 
Sharing privacy-correlated information in return for useful service has become a common practice in modern society. For example, users may trade in location information for the localization service, browsing history for the recommendation service, and daily activity information for the health monitoring service. Despite the convenience and benefits brought by the various services, directly sharing privacy-correlated information may result in unwanted privacy leakage, e.g., home address or political affiliation. Therefore, it is of paramount importance to develop information disclosure methodologies that balance between the privacy loss and a desired level of utility.

In this paper, we adopt a statistical inference framework proposed by Calmon \emph{et al.} in \cite{FduPC-NF:12} and study the impact of the adversary's prior information on the privacy-utility trade-offs. The framework in \cite{FduPC-NF:12} is outlined in Fig.~\ref{fig:framework} and described as follows. A user has some useful information $Y$ that she/he wants to share with a service provider to gain some utility. However, the useful information $Y$ to be shared is correlated with the user's private information $X$ through a joint distribution $p_{X,Y}$, and the service provider could infer $X$ based on the shared information. To reduce the information leakage about $X$, the user instead shares the perturbed information $Z$ produced by a privacy mapping $p_{Z|Y}$, and receives a possibly lower utility based on $Z$. The privacy mapping $p_{Z|Y}$ is a design variable that simultaneously controls the distortion between $Y$ and $Z$ and the information leakage about $X$ from $Z$, and thus it determines the privacy-utility trade-offs. In the above framework, the service provider knows precisely the statistical correlation $p_{X,Y}$. However, this omniscience assumption may not hold in practice and may result in a potentially conservative design of the privacy mapping $p_{Z|Y}$. In this paper, we focus on the scenario where the service provider knows a biased correlation $\hat{p}_{X,Y}\neq p_{X,Y}$ and investigate the implications of this information asymmetry. We use a limited adversary to refer to a service provider that knows a biased correlation.

\begin{figure}
	\centering
	\begin{tikzpicture}[auto, thick, node distance=2cm, >=triangle 45]
		\draw
		node[align=center] at (0,0)[block] (priv) {Private $X$}
		node[align=center] at (3.3,0)[block] (useful) {Useful $Y$}
		node[align=center] at (6.6,0)[block] (shared) {Released $Z$}
		node[align=center,minimum size=0pt] at (4.95,-0.66) (output1) {}		
		node[align=center,minimum size=0pt] at (4.95,-1.8) (output2) {}
		;					
		
		\draw[<->](priv) -- node{coupled} (useful);
		
		\draw[->](useful) -- node[above] {share}  node[below] {$p_{Z|Y}$} (shared);
		
		\draw[->] (shared) -- (6.6,1.1) -| node[label={[xshift=3cm, yshift=-0.2cm]inference ({\color{red}privacy loss})}] {}(priv);
		
		\draw[-] (shared) -- (6.6,-0.8) -| node[] {}(useful);
		
		\node [fit=(useful) (priv),draw,dashed,black] {};			
		%\node[fit=(shared),draw,dashed,black]{};
				
		\draw[->](output1) -- node[auto=left,align=center] {distortion\\ ({\color{red}utility loss})} (output2);		
	\end{tikzpicture}
\vspace{-10pt}
\caption{A statistical inference framework where the probabilistic privacy mapping $p_{Z|Y}$ perturbs the useful information $Y$ to achieve a trade-off between the privacy loss (inference of $X$ based on $Z$) and the utility loss (distortion between $Y$ and $Z$).}\label{fig:framework}
\end{figure}
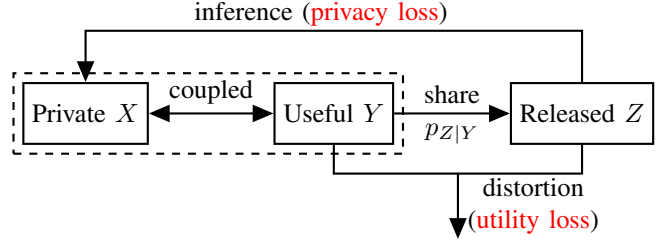

\paragraph*{Literature review} 
Various metrics that quantify privacy leakage in different scenarios and applications exist in the literature. Differential privacy is a popular and widely studied privacy notion that protects the privacy of individual records from queries of databases~\cite{CD-AR:14}. A differentially private mechanism ensures that a single entry change in the database does not incur significant changes in the output distribution by returning a randomized answer. In the control community, differential privacy has been adapted and applied in many privacy-critical problems such as  filtering~\cite{JLeN-GJP:14}, multi-agent consensus~\cite{EN-PT-JC:17}, and distributed optimization~\cite{SH-UT-GJP:17,EN-PT-JC:18, XC-JZ-VHP-ZT:20,TD-SZ-JH-CC-XG:21}; see
\cite{JC-GED-SH-JLeN-SM-GJP:16, SH-GJP:18,YL-MZ:19} for comprehensive surveys on privacy in systems and control. Since differential privacy does not rely on the distribution of the user data, it provides the ``worst-case" privacy guarantees~\cite{WW-LY-JZ:16}.

Different from differential privacy, information-theoretic privacy measures such as mutual information~\cite{FduPC-NF:12, LS-SRR-HVP:13, MPJ-LZ-SC:18}, maximal leakage~\cite{II-ABW-SK:20}, maximal $\alpha$-leakage~\cite{JL-OK-LS-FduPC:19}, and total variation distance~\cite{BR-DG:20}, take into account the prior data distribution; see~\cite{IW-DE:18, MB-OG-AY-FO-HVP-LS-RFS:21} for overviews on information-theoretic privacy and security. Since the privacy and utility requirements usually compete,  appropriate privacy mappings need to be designed to achieve a trade-off between the two types of requirements. The work in~\cite{FduPC-NF:12} proposes the framework in Fig.~\ref{fig:framework} to study privacy-utility trade-offs. The authors design a privacy mapping $p_{Z|Y}$ by solving a convex program such that the mutual information between $X$ and $Z$ is minimized and the average distortion (measured by a distortion function) between $Y$ and $Z$ is constrained. The framework is extended to scenarios where $X$ and $Y$ are time sequences in~\cite{MAE-NF:15} and where data availability differs for the design of the privacy mapping (e.g., $p_{Z|X,Y}$) in~\cite{YOB-YW-PI:16}. The work in~\cite{SS-AZ-FduPC-SB-NF-BK-PO-NT:15}  deals with the case when the true correlation $p_{X,Y}$ may not be known to the user, and a possibly mismatched correlation is used in the privacy mapping design. The authors also study the quantization problem to cope with design variables in high dimensions. The recent work~\cite{MD-HW-FduPC-LS:20} discusses the situation where there is a discrepancy between the empirical correlation used in the design of the privacy mapping and the true correlation in practice, and the authors show that the privacy mapping asymptotically converges to the optimal one as the sample size increases for various privacy metrics. In~\cite{SA-FA-TL:14} and \cite{FduPC-AM-MM:15}, the mutual information serves as both the privacy metric and the utility function in the privacy mapping design. Information-theoretic privacy measures have also appeared in various applications. The authors in~\cite{JL-LS-VYFT-FduPC:18} use mutual information as the privacy metric and design the optimal privacy mapping for hypothesis testing; leveraging a rechargeable battery in households, Li \emph{et al.} in~\cite{SL-AK-AM:18}  study the optimal battery charging policy that minimizes the information leakage, measured by the normalized mutual information, to the utility provider; Nekouei \emph{et al.} formulate the privacy-aware estimation problem in~\cite{EN-HS-MS-KHJ:21}, where the authors build an optimal estimator of a public random variable under a constraint on the privacy level of a correlated private variable. See~\cite{EN-TT-MS-KHJ:19} for more applications of information-theoretic notions in estimation and control.

In the previous works regarding information-theoretic privacy, although the employed metrics are information-theoretically well-posed and meaningful, they all have an implicit assumption regarding the capability of the adversary, i.e., the adversary has the \emph{same} statistical information as the user. However, such a worst-case assumption may not be valid in practice, and one could (and should) exploit the possibly imperfect information of the adversary to achieve improved privacy guarantees. In this paper, we relax the assumption that the adversary knows the precise statistical correlation between the private and useful information and investigate the implications of this relaxation on the privacy-utility trade-offs.

\paragraph*{Contributions}
In this paper, we study the impact of the adversary's information on the privacy-utility trade-offs under a statistical inference framework. We show that the information asymmetry brings advantages to the user and leads to higher inference costs for the adversary. The main contributions of this paper are as follows.
\begin{enumerate}
\item We first identify necessary and sufficient conditions on the adversary's information so that the inference process and the inference costs for the adversary are well-posed. Given these conditions, we show that the inference costs for the adversary increase as a result of the information asymmetry between the user and the adversary.
\item We formulate the design problem of the privacy mapping as a \emph{difference of convex functions program} and adopt the concave-convex procedure to solve it. Moreover, we derive a sufficient condition on the adversary's information under which the design problem is convex. 
\item When the adversary's biased information is not precisely available, we take a Bayesian approach and assume a distribution over the information the adversary may have. Since the expected cost for the adversary in this case does not admit a closed-form expression, we derive a lower bound for it and then design a privacy mapping that maximizes the obtained lower bound.
\end{enumerate}

\paragraph*{Organization} We organize the rest of the paper as follows. Section~\ref{sec:prelim} reviews relevant information-theoretic concepts. We introduce the original privacy-utility trade-off problem and the problem of interest in Section~\ref{sec:tradeoff}. We then study the impact of the adversary's information on the design of the privacy mapping in detail in Section~\ref{sec:known}. Section~\ref{sec:expected} presents a Bayesian approach for the case when the adversary's information is not exactly available. We provide a numerical example using a census data set in Section~\ref{sec:numerics}. Section~\ref{sec:conclusion} concludes the paper.

\section{Notation and Preliminaries}\label{sec:prelim}
\subsection{Notation} 

Let $\mathbb{R}$, $\mathbb{R}^n$, and $\mathbb{R}^{m\times n}$ be the set of real numbers, set of real vectors of dimension $n$, and set of real matrices of dimension $m\times n$, respectively. We use bold symbols to denote vectors and matrices, and we use capital letters to denote random variables. All vectors in this paper are column vectors. We denote the probability simplex in dimension $n$ by $\Delta_n$, i.e., $\Delta_n=\{\mathbf{x}\in\mathbb{R}^n\,|\,\sum_{i=1}^n\mathbf{x}_i=1,\mathbf{x}_i\geq0\textup{ for } 1\leq i\leq n\}$. The matrix (vector) of $1$'s in dimension $m\times n$ ($n$) is $\mathbf{1}_{m\times n}$ ($\mathbf{1}_n$). We use $\mathbf{A}\odot\mathbf{B}$ ($\mathbf{A}\oslash\mathbf{B}$) to denote the component-wise product (division, assuming well-defined) of two matrices $\mathbf{A}\in\mathbb{R}^{m\times n}$ and $\mathbf{B}\in\mathbb{R}^{m\times n}$. For a matrix $\mathbf{A}\in\mathbb{R}^{m\times n}$, $\mathbf{A}_{i,*}$ and $\mathbf{A}_{*,j}$ are column vectors that represent the $i$-th row and $j$-th column of $\mathbf{A}$, respectively, for $i\in\{1\dots,m\}$ and $j\in\{1,\dots,n\}$. We denote the Frobenius norm of a matrix $\mathbf{A}\in\mathbb{R}^{m\times n}$ by $\|\mathbf{A}\|_F$. For a finite set $\mathcal{S}$, $|\mathcal{S}|$ is its cardinality.

\subsection{Entropy, conditional entropy and mutual information}
For a discrete random variable $X$ over $\mathcal{X}=\{1,\cdots,|\mathcal{X}|\}$ with the probability mass function $p_X:\mathcal{X}\to[0,1]$, the \emph{entropy} $H(X)$ of $X$ is defined by 
\begin{equation*}
H(X) =-\sum_{x\in\mathcal{X}}p_X(x)\log p_X(x),
\end{equation*}
where the logarithm is the natural logarithm for the ease of exposition and $0\log0=0$. The entropy $H(X)$ of a random variable $X$ measures the amount of information (or uncertainty) $X$ contains. For a pair of discrete random variables $X$ and $Y$ taking values in $\mathcal{X}=\{1,\cdots,|\mathcal{X}|\}$ and $\mathcal{Y}=\{1,\cdots,|\mathcal{Y}|\}$, respectively, the \emph{conditional entropy} of $X$ given $Y$ is defined by
\begin{align*}
H(X|Y)&=\sum_{y\in\mathcal{Y}}p_Y(y)H(X|Y=y)\\
&=-\sum_{y\in\mathcal{Y}}p_Y(y)\sum_{x\in\mathcal{X}}p_{X|Y}(x|y)\log p_{X|Y}(x|y),
\end{align*}
where $p_{X|Y}$ is the conditional distribution of $X$ given $Y$. The conditional entropy $H(X|Y)$ measures the amount of information in $X$ provided that $Y$ is given. In particular, if $X$ and $Y$ are independent, then we have $H(X|Y)=H(X)$. The \emph{mutual information} $I(X;Y)$ of the random variables $X$ and $Y$ is defined by
\begin{equation*}
I(X;Y)=H(X)-H(X|Y)=H(Y)-H(Y|X),
\end{equation*}
which measures the uncertainty reduction of the random variable $X$ when $Y$ is given. The \emph{Kullback–Leibler} (KL) \emph{divergence} $\mathcal{KL}(p_X||q_X)$ of two probability mass functions $p_X$ and $q_X$ over $\mathcal{X}$ is defined by
\begin{equation*}
\mathcal{KL}(p_X||q_X)=\sum_{x\in\mathcal{X}}p_X(x)\log\frac{p_X(x)}{q_X(x)}.
\end{equation*}
The KL divergence satisfies $\mathcal{KL}(p_X||q_X)\geq0$ with equality if and only if $p_X(x)=q_X(x)$ for all $x\in\mathcal{X}$ \cite[Theorem 2.6.3]{TMC-JAT:12}.

\section{Privacy-utility trade-offs and problem of interest}\label{sec:tradeoff}
This section presents the original privacy-utility trade-off problem in~\cite{FduPC-NF:12} and states the problem of interest.

\subsection{Privacy-utility trade-offs}
A user has some private information $X$, e.g., political affiliation, and some useful information $Y$ correlated with $X$ through $p_{X,Y}$ , e.g., media preferences. The user could share the useful information $Y$ with a service provider to obtain some utility, e.g., content recommendations. However, since $X$ and $Y$ are correlated, the (adversarial) service provider might infer $X$ from $Y$. Therefore, the user instead shares the carefully perturbed information $Z$ of $Y$ via a privacy mapping $p_{Z|Y}$, which makes $X$ less inferable but maintains good usability of $Z$.
 
For a realization $x\in\mathcal{X}$ of the private information $X$ and an adversary's perceived prior $q_X\in\Delta_{|\mathcal{X}|}$, the \emph{inference cost} for the adversary is given by $C(x,q_X)$, which under the log-loss function~\cite{NM-MF:98,FduPC-NF:12,AM-SS-NF-MM:14}  becomes $C(x,q_X)=-\log q_X(x)$. Since $X$ is a random variable with the probability mass function $p_{X}$, the expected cost for the adversary is
\begin{equation}\label{eq:adcost}
	c_0(q_X)=\sum_{x\in\mathcal{X}}p_X(x)C(x,q_X)=-\sum_{x\in\mathcal{X}}p_X(x)\log q_X(x).
\end{equation}
If no additional information is available, then the adversary selects $q_X^*$ so as to minimize the expected cost, i.e.,
\begin{equation}\label{eq:optq}
q_X^*=\argmin_{q_X}c_0(q_X)=p_X,
\end{equation}
with the associated optimal cost
\begin{equation*}
c_0^*=\min_{q_X}c_0(q_X)=H(X).
\end{equation*}
When a realization $z\in\mathcal{Z}$ of the perturbed information $Z$ (correlated with $Y$ and thus $X$) is disclosed according to a privacy mapping $p_{Z|Y}$, the adversary computes the posterior $p_{X|Z}$ and the associated cost
\begin{equation}\label{eq:conditionalcost}
	c_z(q_X)=-\sum_{x\in\mathcal{X}}p_{X|Z}(x|z)\log q_{X,z}(x).
\end{equation}
In this case, the optimal cost $c_z^*$ and associated distribution $q_{X,z}^*$ become
\begin{equation}\label{eq:optqwithz}
c_z^*=H(X|Z=z)\quad{\textup{and}}\quad q_{X,z}^*=p_{X|Z}(\cdot|z).
\end{equation}
Finally, the average cost for the adversary given the information $Z$ can be computed by
\begin{equation*}
c_Z=\E_Z[c_z^*]=H(X|Z).
\end{equation*}
The cost reduction for the adversary, or the information leakage for the user, due to the information release, is then
\begin{equation*}
L=c_0^*-c_Z=I(X;Z)=\mathcal{KL}(p_{X,Z}||{p}_{X}{p}_{Z}).
\end{equation*}
The privacy-utility trade-off problem concerns the design of the privacy mapping $p_{Z|Y}$ such that the information leakage $L$ is minimized, or equivalently, the conditional entropy $H(X|Z)$ is maximized (since $H(X)$ is a constant). On the other hand, the perturbation to the useful information $Y$ induces a utility loss for the user measured by $\E_{Y,Z}[d(y,z)]$, where $d(y,z):\mathcal{Y}\times\mathcal{Z}\to\mathbb{R}_{\geq0}$ is a distortion function indicating how far a realization $z\in\mathcal{Z}$ of the released information $Z$ is away from a realization $y\in\mathcal{Y}$ of the useful information $Y$ and $d(y,z)=0$ for $y=z$. To balance the privacy and utility losses, we solve the following optimization problem~\cite{FduPC-NF:12}
\begin{subequations}\label{eq:utilityprivacy}
	\begin{align}
	\minimize_{p_{Z|Y}} \quad & L\label{eq:obj}\\
	\st \quad & \E_{Y,Z}[d(y,z)] \leq\delta,\label{eq:constraintdis}\\
	&p_{Z|Y}(\cdot|y)\in\Delta_{|\mathcal{Z}|},\quad\forall y\in\mathcal{Y},\label{eq:constraintconv}
\end{align}
\end{subequations}
where $\delta\geq0$ is the tolerance for the utility loss. 

\subsection{Problem of interest}
In order to determine the optimal prior distribution $q_X^*$ in~\eqref{eq:optq} and the optimal posterior distribution $q_{X,z}^*$ in~\eqref{eq:optqwithz}, the adversary needs to know precisely the correlation $p_{X,Y}$ between the private information $X$ and useful information $Y$. In this paper, we study cases where the adversary has imperfect information about the correlation $p_{X,Y}$, and we design a privacy mapping $p_{Z|Y}$ that achieves better privacy-utility trade-offs than the case of perfect information for the adversary. Without loss of generality, we make the following assumption on the correlation $p_{X,Y}$.
\begin{assumption}[Correlation between private and useful information]\label{assump:corpri}
The marginal distributions $p_X$ and $p_Y$ of the joint distribution $p_{X,Y}$ satisfy $p_X(x)>0$  and $p_Y(y)>0$ for all $x\in\mathcal{X}$ and $y\in\mathcal{Y}$, respectively. 
\end{assumption}
Assumption~\ref{assump:corpri} essentially requires that no redundant elements that have zero probability of being realized are included in the sets $\mathcal{X}$ and $\mathcal{Y}$. To make~\eqref{eq:constraintdis}-\eqref{eq:constraintconv} feasible for any $\delta\geq0$, we also assume that it is always possible to release any useful information $y\in\mathcal{Y}$ directly, i.e., we assume $\mathcal{Y}\subset\mathcal{Z}$.

\section{A limited adversary with known correlations}\label{sec:known}
In Section~\ref{sec:tradeoff}, we derived the costs for the adversary under the assumption that it can make informed decisions, i.e., the adversary knows exactly the correlation $p_{X,Y}$ between the private and useful information. Despite being able to account for the worst-case scenario, such an omniscience assumption might lead to a conservative design of the privacy mapping $p_{Z|Y}$. Moreover, it is often the case that only limited information is available to the adversary in practice. In this section, we study cases in which the adversary has imperfect information about $p_{X,Y}$ and their implications. In particular, we show that it is possible to exploit such information asymmetry between the user and the adversary to design a privacy mapping that achieves improved privacy guarantees compared to the case in which the adversary knows precisely $p_{X,Y}$. We focus on the impact of the adversary's biased information about $p_{X,Y}$ and assume  that the privacy mapping $p_{Z|Y}$ is publicly available \cite{WW-LY-JZ:16}. We will denote all variables related to the adversary by symbols with hats. 

\subsection{Assumptions, costs, and problem of interest}
When the adversary's information $\hat{p}_{X,Y}$ about the correlation between $X$ and $Y$ is imprecise, i.e., $\hat{p}_{X,Y}\neq p_{X,Y}$, the costs derived in Section~\ref{sec:tradeoff} are not valid or even well-defined. In fact, it is possible that given a privacy mapping $p_{Z|Y}$, a realization $z\in\mathcal{Z}$ can be generated with positive probability under $p_{X,Y}$, but zero probability under $\hat{p}_{X,Y}$. As a result, the posterior distribution $\hat{p}_{X|Z}$ calculated by the adversary is ill-defined. Hereafter, we identify minimal appropriate assumptions on $\hat{p}_{X,Y}$ so that important objects such as $\hat{p}_{X|Z}$ are well-posed and the privacy-utility trade-off problem is meaningful.

%The occurrence of such ``shocking'' events, which is the consequence of the mismatch between  $p_{X,Y}$ and $\hat{p}_{X,Y}$, will make the adversary aware of the information bias and likely adjust the distribution $\hat{p}_{X,Y}$ accordingly. 

\begin{assumption}[Support of $\hat{p}_{X,Y}$]\label{assump:supportad}
Given a joint distribution $p_{X,Y}$ that satisfies Assumption~\ref{assump:corpri}, the joint distribution $\hat{p}_{X,Y}$ satisfies 
\begin{enumerate}[label=(\alph*)]
	\item\label{assump:ymarginal} the marginal distribution $\hat{p}_Y$ of $\hat{p}_{X,Y}$ has support $\mathcal{Y}$;
	\item\label{assump:xmarginal} the marginal distribution $\hat{p}_X$ of $\hat{p}_{X,Y}$ has support $\mathcal{X}$;
	\item\label{assump:joint} for any $x\in\mathcal{X}$ and $y\in\mathcal{Y}$, if $p_{X,Y}(x,y)>0$, then $\hat{p}_{X,Y}(x,y)>0$.  
\end{enumerate}
\end{assumption}

In Assumption~\ref{assump:supportad}, \ref{assump:supportad}\ref{assump:joint} is the strongest and implies the other two. As we shall see, it is also the minimal assumption required to make the privacy-utility trade-off problem nontrivial. Assumptions~\ref{assump:supportad}\ref{assump:ymarginal} and \ref{assump:supportad}\ref{assump:xmarginal} are intermediate assumptions that guarantee well-posedness of important quantities that will be used later. We first show that Assumption~\ref{assump:supportad}\ref{assump:ymarginal} is a necessary and sufficient condition for the posterior $\hat{p}_{X|Z}$ to be well-defined under any privacy mapping $p_{Z|Y}$.

\begin{lema}[Well-posedness of the posterior $\hat{p}_{X|Z}$]\label{lemma:wellposedness}
Given a joint distribution $p_{X,Y}$ that satisfies Assumption~\ref{assump:corpri}, the posterior $\hat{p}_{X|Z}$ is well-defined under any privacy mapping $p_{Z|Y}$ if and only if $\hat{p}_{X,Y}$ satisfies Assumption~\ref{assump:supportad}\ref{assump:ymarginal}.
\end{lema}
\begin{proof}
For any $x\in\mathcal{X}$ and $z\in\mathcal{Z}$, the posterior $\hat{p}_{X|Z}(x|z)$ can be written as
\begin{equation}\label{eq:posterior}
	\hat{p}_{X|Z}(x|z)=\frac{\hat{p}_{X,Z}(x,z)}{\hat{p}_{Z}(z)}.
\end{equation}
For any $z\in\mathcal{Z}$ such that 
\begin{equation}\label{eq:pz}
	p_Z(z)=\sum_{y\in\mathcal{Y}}p_{Y}(y)p_{Z|Y}(z|y)>0,
\end{equation}
\eqref{eq:posterior} is well-defined if and only if
\begin{equation}\label{eq:Zad}
	\hat{p}_{Z}(z)=\sum_{y\in\mathcal{Y}}\hat{p}_{Y}(y)p_{Z|Y}(z|y)>0.
\end{equation}

\emph{Necessity} (Well-posedness$\implies$ Assumption~\ref{assump:supportad}\ref{assump:ymarginal}): We prove by contrapositive. Suppose that $\hat{p}_{X,Y}$ does not satisfy Assumption~\ref{assump:supportad}\ref{assump:ymarginal}, i.e., there exists a $y'\in\mathcal{Y}$ such that $\hat{p}_Y(y')=0$. Then, for any $z\in\mathcal{Z}$ and the privacy mapping $p_{Z|Y}$ that satisfies
\begin{equation*}
p_{Z|Y}(z|y')>0\quad\textup{and}\quad p_{Z|Y}(z|y'')=0,~\forall y''\neq y',
\end{equation*}
we have that
 \begin{equation*}
	p_Z(z)=\sum_{y\in\mathcal{Y}}p_{Y}(y)p_{Z|Y}(z|y)=p_{Y}(y')p_{Z|Y}(z|y')>0.
\end{equation*}
However,
\begin{equation*}
\hat{p}_{Z}(z)=\sum_{y\in\mathcal{Y}}\hat{p}_{Y}(y)p_{Z|Y}(z|y)=\hat{p}_{Y}(y')p_{Z|Y}(z|y')=0,
\end{equation*}
which implies that $\hat{P}_{X|Z}$ is not well-defined.

\emph{Sufficiency} (Assumption~\ref{assump:supportad}\ref{assump:ymarginal}$\implies$ Well-posedness): Suppose that $\hat{p}_{X,Y}$ satisfies Assumption~\ref{assump:supportad}\ref{assump:ymarginal}. Then, for any privacy mapping $p_{Z|Y}$ and $z\in\mathcal{Z}$ that satisfies \eqref{eq:pz}, we have that there exists at least one $y'\in\mathcal{Y}$ such that $p_Y(y')p_{Z|Y}(z|y')>0$. Therefore, by Assumption~\ref{assump:supportad}\ref{assump:ymarginal}, we also have $\hat{p}_Y(y')p_{Z|Y}(z|y')>0$ , which implies~\eqref{eq:Zad}.
\end{proof}

With Lemma~\ref{lemma:wellposedness}, we are now ready to calculate the costs for the adversary who knows a biased distribution $\hat{p}_{X,Y}\neq p_{X,Y}$.

\begin{thom}[Costs for a limited adversary]\label{thm:cost}
Given joint distributions $p_{X,Y}$ and $\hat{p}_{X,Y}$ that satisfy Assumption~\ref{assump:corpri} and~\ref{assump:supportad}\ref{assump:ymarginal}, respectively, the average costs for the adversary before and after the release of $Z$ are given by
\begin{align}
&	{\hat{c}_0}^*=H(X)+\mathcal{KL}(p_X||\hat{p}_X),\label{eq:costimperfect_0}\\
&	\hat{c}_Z=H(X|Z)+\mathcal{KL}(p_{X|Z}||\hat{p}_{X|Z}).\label{eq:costimperfect_z}
\end{align}
Moreover, the information leakage satisfies
\begin{equation}\label{eq:leakageimperfect}
\hat{L}={\hat{c}_0}^*-\hat{c}_Z=\mathcal{KL}(p_{X,Z}||\hat{p}_{X}\hat{p}_{Z})-\mathcal{KL}(p_{X,Z}||\hat{p}_{X,Z}).
\end{equation}
\end{thom}
\begin{proof}
Since the adversary's information about the distribution of the private information $X$ is the marginal distribution $\hat{p}_X$ of $\hat{p}_{X,Y}$,  the expected cost in~\eqref{eq:adcost} from the adversary's perspective becomes
\begin{equation}\label{eq:adcostimperfect}
	\hat{c}_0(\hat{q}_X)=\sum_{x\in\mathcal{X}}\hat{p}_{X}(x)C(x,\hat{q}_X)=-\sum_{x\in\mathcal{X}}\hat{p}_{X}(x)\log\hat{q}_X(x),
\end{equation}
which has the optimizer ${\hat{q}_{X}}^*=\hat{p}_X$.
From the user's perspective, when the adversary uses  ${\hat{q}_{X}}^*$, the actual cost for the adversary is
\begin{align*}
{\hat{c}_0}^*=\sum_{x\in\mathcal{X}}p_{X}(x)C(x,{\hat{q}_{X}}^*)&=-\sum_{x\in\mathcal{X}}p_{X}(x)\log \hat{p}_X(x)\\
&=H(X)+\mathcal{KL}(p_X||\hat{p}_X).
\end{align*}

Similarly, when a realization $z\in\mathcal{Z}$ is released, the adversary computes  a biased posterior $\hat{p}_{X|Z}$, which is well-defined by Assumption~\ref{assump:supportad}\ref{assump:ymarginal}, based on the correlation $\hat{p}_{X,Y}$ and uses the corresponding optimizer $\hat{q}_{X,z}^*=\hat{p}_{X|Z}(\cdot|z)$. From the user's perspective, the actual cost for the adversary in this case is
\begin{align*}
\hat{c}_z^*&=\sum_{x\in\mathcal{X}}p_{X|Z}(x|z)C(x,\hat{q}_{X,z}^*(x))\\
&=H(X|Z=z)+\sum_{x\in\mathcal{X}}p_{X|Z}(x|z)\log\frac{p_{X|Z}(x|z)}{\hat{p}_{X|Z}(x|z)},%\mathcal{KL}(p_{X|Z}||\hat{p}_{X|Z}),
\end{align*}
and the average cost for the adversary given the information $Z$ can be computed by
\begin{align*}
	\hat{c}_Z=\E_{Z}[\hat{c}_z^*]&=H(X|Z)+\mathcal{KL}(p_{X|Z}||\hat{p}_{X|Z}).
\end{align*}

Finally, the information leakage, defined as the difference between the costs for the adversary before and after releasing $Z$, is
\begin{align}\label{eq:leakgead}
	\begin{split}
	&\quad\quad\hat{L}\\
	&={\hat{c}_0}^*-\hat{c}_Z\\
&=-\sum_{x\in\mathcal{X}}p_{X}(x)\log \hat{p}_X(x)+\sum_{x,z}p_{X,Z}(x,z)\log \hat{p}_{X|Z}(x|z)\\
&=\sum_{x,z}p_{X,Z}(x,z)\log \frac{\hat{p}_{X,Z}(x,z)}{ \hat{p}_X(x)\hat{p}_Z(z)}\\
&=\sum_{x,z}p_{X,Z}(x,z)\log \frac{p_{X,Z}(x,z)}{ \hat{p}_X(x)\hat{p}_Z(z)}\\
&\quad+\sum_{x,z}p_{X,Z}(x,z)\log \frac{\hat{p}_{X,Z}(x,z)}{p_{X,Z}(x,z)}\\
&=\mathcal{KL}(p_{X,Z}||\hat{p}_{X}\hat{p}_{Z})-\mathcal{KL}(p_{X,Z}||\hat{p}_{X,Z}).
	\end{split}
\end{align}
\end{proof}

From~\eqref{eq:costimperfect_0} and~\eqref{eq:costimperfect_z} in Theorem~\ref{thm:cost}, we observe that the costs for the limited adversary are higher than the respective costs for the adversary who has perfect information in Section~\ref{sec:tradeoff}. The differences in these costs depend explicitly on the biased information $\hat{p}_{X,Y}$, and as expected, when $\hat{p}_{X,Y}=p_{X,Y}$, we recover the costs introduced in  Section~\ref{sec:tradeoff}. 

To formulate the problem of interest similar to~\eqref{eq:obj}-\eqref{eq:constraintconv}, we need to further impose assumptions on $\hat{p}_{X,Y}$ so that the information leakage $\hat{L}$ is bounded below.

\begin{lema}[Finite costs]\label{lemma:infinitecosts}
Given joint distributions $p_{X,Y}$ and $\hat{p}_{X,Y}$ that satisfy Assumption~\ref{assump:corpri} and~\ref{assump:supportad}\ref{assump:ymarginal}, respectively, the following statements hold:
\begin{enumerate}[label=(\roman*)]
\item\label{itm:cost0} the average cost $\hat{c}^*_0$ in~\eqref{eq:costimperfect_0} is finite if and only if $\hat{p}_{X,Y}$ further satisfies Assumption~\ref{assump:supportad}\ref{assump:xmarginal};
\item\label{itm:costz} the average cost $\hat{c}_Z$ in~\eqref{eq:costimperfect_z} is finite under any privacy mapping $p_{Z|Y}(z|y)$ if and only if $\hat{p}_{X,Y}$ further satisfies Assumption~\ref{assump:supportad}\ref{assump:joint}.
\item\label{itm:leakage} the information leakage $\hat{L}$ in~\eqref{eq:leakageimperfect} is well-defined and finite under any privacy mapping $p_{Z|Y}$ if and only if $\hat{p}_{X,Y}$ further satisfies Assumption~\ref{assump:supportad}\ref{assump:joint}.
\end{enumerate}
\end{lema}
\begin{proof}
Regarding~\ref{itm:cost0}, it follows directly from~\eqref{eq:costimperfect_0} and the definition of KL divergence.

Regarding~\ref{itm:costz}, since the conditional entropy $H(X|Z)$ in~\eqref{eq:costimperfect_z} satisfies $H(X|Z)\leq H(X)$ and is always finite, $\hat{c}_Z$ is finite if and only if $\mathcal{KL}(p_{X|Z}||\hat{p}_{X|Z})$ is finite. We expand $\mathcal{KL}(p_{X|Z}||\hat{p}_{X|Z})$ as
\begin{align}\label{eq:expandKL}
	\begin{split}
		\mathcal{KL}(p_{X|Z}||\hat{p}_{X|Z})&=\sum_{x,z}p_{X,Z}(x,z)\log\frac{p_{X|Z}(x|z)}{\hat{p}_{X|Z}(x|z)}\\
		&=\underbrace{\sum_{x,z}p_{X,Z}(x,z)\log\frac{p_{X,Z}(x,z)}{p_Z(z)}}_\text{first term}\\
		&\quad-\underbrace{\sum_{x,z}{p}_{X,Z}(x,z)\log\frac{\hat{p}_{X,Z}(x,z)}{\hat{p}_Z(z)}}_\text{second term},
	\end{split}
\end{align}
where we only care about $z\in\mathcal{Z}$ such that $p_{Z}(z)\neq0$ and $\hat{p}_{Z}(z)\neq0$ (otherwise the posteriors $p_{X|Z}$ and $\hat{p}_{X|Z}$ will not be calculated). Note that the first term in~\eqref{eq:expandKL} is always finite, and therefore $\hat{c}_Z$ is finite if and only if the second term in~\eqref{eq:expandKL} is finite. We further expand the second term in~\eqref{eq:expandKL} and show its explicit dependence on the privacy mapping $p_{Z|Y}$ as
\begin{multline}\label{eq:secondterm}
\sum_{x,z}{p}_{X,Z}(x,z)\log\frac{\hat{p}_{X,Z}(x,z)}{\hat{p}_Z(z)}\\=\sum_{x,y,z}{p}_{X,Y}(x,y){p}_{Z|Y}(z|y)\\
	\cdot\log\frac{\sum_{y\in\mathcal{Y}}\hat{p}_{X,Y}(x,y){p}_{Z|Y}(z|y)}{\hat{p}_Z(z)}.
\end{multline}

\emph{Necessity} (finiteness$\implies$Assumption~\ref{assump:supportad}\ref{assump:joint}):  We prove by contrapositive. Suppose that $\hat{p}_{X,Y}$ does not satisfy Assumption~\ref{assump:supportad}\ref{assump:joint}, i.e., there exist $x'\in\mathcal{X}$ and $y'\in\mathcal{Y}$ such that ${p}_{X,Y}(x',y')>0$ and $\hat{p}_{X,Y}(x',y')=0$. For a $z'\in\mathcal{Z}$, let the privacy mapping satisfy $p_{Z|Y}(z'|y')>0$ and $p_{Z|Y}(z'|y'')=0$ for $y''\neq y'$, then
\begin{equation}\label{eq:negativeinf1}
	\sum_{y\in\mathcal{Y}}{p}_{X,Y}(x',y){p}_{Z|Y}(z'|y)={p}_{X,Y}(x',y'){p}_{Z|Y}(z'|y')>0,
\end{equation}
and
\begin{equation}\label{eq:negativeinf2}
	\sum_{y\in\mathcal{Y}}\hat{p}_{X,Y}(x',y'){p}_{Z|Y}(z'|y')=\hat{p}_{X,Y}(x',y'){p}_{Z|Y}(z'|y')=0,
\end{equation}
in which case \eqref{eq:secondterm}  is negative infinite.

\emph{Sufficiency} (Assumption~\ref{assump:supportad}\ref{assump:joint}$\implies$finiteness): Suppose $\hat{p}_{X,Y}$ satisfies Assumption~\ref{assump:supportad}\ref{assump:joint}. If for some $x'\in\mathcal{X}$ and $z'\in\mathcal{Z}$, we have that
\begin{equation}\label{eq:condition2}
	\sum_{y\in\mathcal{Y}}\hat{p}_{X,Y}(x',y){p}_{Z|Y}(z'|y)=0,
\end{equation}
then we must have $\hat{p}_{X,Y}(x',y){p}_{Z|Y}(z'|y)=0$ for any $y\in\mathcal{Y}$. Assumption~\ref{assump:supportad}\ref{assump:joint} ensures that ${p}_{X,Y}(x',y){p}_{Z|Y}(z'|y)=0$ for any $y\in\mathcal{Y}$. Therefore, \eqref{eq:secondterm} is finite.

Regarding~\ref{itm:leakage}, \emph{Necessity} (Well-posedness $\implies$Assumption~\ref{assump:supportad}\ref{assump:joint}): We prove by contrapositive. Suppose $\hat{p}_{X,Y}$ does not satisfy Assumption~\ref{assump:supportad}\ref{assump:joint}. We further consider two scenarios
\begin{enumerate}
\item if $\hat{p}_{X,Y}$ satisfies Assumption~\ref{assump:supportad}\ref{assump:xmarginal}, then by~\ref{itm:cost0} and \ref{itm:costz}, $\hat{c}_Z$ is infinite and $\hat{c}_0^*$ is finite, and thus $\hat{L}$ is infinite;

\item if $\hat{p}_{X,Y}$ does not satisfy Assumption~\ref{assump:supportad}\ref{assump:xmarginal}, i.e., there exists at least one $x\in\mathcal{X}$ such that $\hat{p}_X(x)=0$, then from~\eqref{eq:leakgead} we have
\begin{equation}
\hat{L}=\sum_{x,z}p_{X,Z}(x,z)\log \frac{\hat{p}_{Z|X}(z|x)}{\hat{p}_Z(z)},
\end{equation}
where $\hat{p}_{Z|X}(z|x)$ is not defined. 
\end{enumerate}

\emph{Sufficiency} (Assumption~\ref{assump:supportad}\ref{assump:joint} $\implies$Well-posedness): if $\hat{p}_{X,Y}$ satisfies Assumption~\ref{assump:supportad}\ref{assump:joint}, then by~\ref{itm:cost0} and \ref{itm:costz}, both $\hat{c}_0^*$ and $\hat{c}_Z$ are finite, and therefore $\hat{L}$ is finite and well-defined.
\end{proof}

In light of the discussions in Lemma~\ref{lemma:infinitecosts}, it is clear that the privacy-utility trade-off problem is only interesting when $\hat{p}_{X,Y}$ satisfies~Assumption~\ref{assump:supportad}\ref{assump:joint}. We formally state the problem of interest as follows.

\begin{problem}[Privacy-utility trade-offs against a limited adversary]\label{prob:privacyutilityad}
Given joint distributions $p_{X,Y}$ and $\hat{p}_{X,Y}$ that satisfy Assumption~\ref{assump:corpri} and~\ref{assump:supportad}\ref{assump:joint}, respectively, find a privacy mapping $p_{Z|Y}$ such that the information leakage $\hat{L}$ is minimized under a utility loss constraint, i.e., solve the following optimization problem
\begin{subequations}\label{eq:utilityprivacyad}
	\begin{align}
		\minimize_{p_{Z|Y}} \quad & \hat{L}\label{eq:objad}\\
		\st \quad & \E_{Y,Z}[d(y,z)] \leq\delta,\label{eq:constraintad}\\
		&p_{Z|Y}(\cdot|y)\in\Delta_{|\mathcal{Z}|},\quad\forall y\in\mathcal{Y}.\label{eq:simplexad}
	\end{align}
\end{subequations}
\end{problem}

\begin{remark}[Optimal solution in the absence of Assumption~\ref{assump:supportad}\ref{assump:joint}]
	In the case when $\hat{p}_{X,Y}$ satisfies Assumption~\ref{assump:supportad}\ref{assump:xmarginal} but not \ref{assump:supportad}\ref{assump:joint}, by Lemma~\ref{lemma:infinitecosts}\ref{itm:costz} and \ref{itm:leakage}, the information leakage $\hat{L}$ can be negative infinity. In fact, as long as the constraint set defined by \eqref{eq:constraintad} and \eqref{eq:simplexad} is nonempty, there always exists a feasible privacy mapping that achieves negative infinite information leakage. Specifically, let $p_{Z|Y}$ be a feasible privacy mapping, and $x'\in\mathcal{X}$ and $y'\in\mathcal{Y}$ be such that $p_{X,Y}(x',y')>0$ and $\hat{p}_{X,Y}(x',y')=0$. We construct a new privacy mapping $\tilde{p}_{Z|Y}$ by following the steps below
	\begin{enumerate}
	\item $\tilde{p}_{Z|Y}(z|y)\gets p_{Z|Y}(z|y)$ for all $y\in\mathcal{Y}$ and $z\in\mathcal{Z}$;
	\item $\tilde{p}_{Z|Y}(y'|y')\gets1$ and $\tilde{p}_{Z|Y}(z|y')\gets0$ for any $z\neq y'$;
	\item $\tilde{p}_{Z|Y}(y''|y'')\gets p_{Z|Y}(y''|y'')+p_{Z|Y}(y'|y'')$ for $y''\neq y'$; 
	\item $\tilde{p}_{Z|Y}(y'|y'')\gets 0$ for $y''\neq y'$.
	\end{enumerate}	
	By construction, $\tilde{p}_{Z|Y}$ is still a valid conditional distribution since the distortion between $Y$ and $Z$ decreases under $\tilde{p}_{Z|Y}$, i.e., $\tilde{p}_{Z|Y}$ satisfies  \eqref{eq:constraintad} and \eqref{eq:simplexad}. On the other hand, by~\eqref{eq:negativeinf1} and~\eqref{eq:negativeinf2}, the information leakage under the modified privacy mapping $\tilde{p}_{Z|Y}$ is negative infinite.
\end{remark}

Unlike the case when the adversary knows perfectly the correlation $p_{X,Y}$, releasing information $Z$ might lead to privacy enhancement against a limited adversary since $Z$ might be misleading from the adversary's perspective. In other words, the information leakage $\hat{L}$ is sign-indefinite and could be negative. On the other hand, although a biased prior leads to higher initial and posterior costs for the adversary as shown in~\eqref{eq:costimperfect_0} and~\eqref{eq:costimperfect_z}, it does not necessarily lead to lower information leakage $\hat{L}$.  The following example illustrates these scenarios. 
\begin{example}[Information leakage against limited adversaries]\label{exam:infoleak}
	 For adversaries with perfect information $p_{X,Y}$ and biased information $\hat{p}_{X,Y}$, the difference in information leakage is
\begin{align}
	\begin{split}\label{eq:diff_info_leak}
L-\hat{L}&=\sum_{x,z}p_{X,Z}(x,z)\log\frac{p_{X,Z}(x,z)}{p_{X}(x)p_Z(z)}\\
&\quad-\sum_{x,z}p_{X,Z}(x,z)\log \frac{\hat{p}_{X,Z}(x,z)}{ \hat{p}_X(x)\hat{p}_Z(z)}\\
&=\mathcal{KL}(p_{X,Z}||\hat{p}_{X,Z})-\mathcal{KL}(p_{X}||\hat{p}_{X})-\mathcal{KL}(p_{Z}||\hat{p}_{Z}).
	\end{split}
\end{align} 
Consider
\begin{center}
	\begin{tabular}{ c|c c }
		$p_{X,Y}$ & $y_1$ & $y_2$ \\ \hline
		$x_1$ & $0.4$ & $0.1$ \\  
		$x_2$ & $0.1$ & $0.4$    
	\end{tabular}\qquad
	\begin{tabular}{ c|c c }
	$p_{Z|Y}$ & $z_1$ & $z_2$ \\ \hline
	$y_1$ & $0.8$ & $0.2$ \\  
	$y_2$ & $0.2$ & $0.8$    
\end{tabular}
\end{center}
and two biased correlations corresponding to two limited adversaries
\begin{center}
\begin{tabular}{ c|c c }
	$\hat{p}_{X,Y}^1$ & $y_1$ & $y_2$ \\ \hline
	$x_1$ & $0.3$ & $0.4$ \\  
	$x_2$ & $0.2$ & $0.1$    
\end{tabular}
\qquad
\begin{tabular}{ c|c c }
	$\hat{p}_{X,Y}^2$ & $y_1$ & $y_2$ \\ \hline
	$x_1$ & $0.3$ & $0.1$ \\  
	$x_2$ & $0.5$ & $0.1$    
\end{tabular}
\end{center}
We calculate the initial costs by~\eqref{eq:costimperfect_0}, the posterior costs by~\eqref{eq:costimperfect_z}, the information leakage by~\eqref{eq:leakageimperfect}, and the differences of information leakage by~\eqref{eq:diff_info_leak}, and we report the results in Table~\ref{tb:costsinfoleak}.
%\begin{table}[http]
%	\centering
%	\begin{tabular}{ c|c c c c}
%		 &$c_0$ ($\hat{c}_0$) & $c_Z$ ($\hat{c}_Z$)& $L$ ($\hat{L}$) &$L-\hat{L}$\\ \hline
%		\textup{Knowing} $p_{X,Y}$ & $0.693$ & $0.627$ &$0.066$&$0$\\  
%		\textup{Knowing} $\hat{p}^1_{X,Y}$ & $0.693$ & $0.844$  &$-0.151$ &$0.217$\\
%		\textup{Knowing} $\hat{p}^2_{X,Y}$ & $0.916$& $0.732$  &$0.184$ &$-0.118$
%	\end{tabular}\caption{Costs and information leakage for different adversaries}\label{tb:costsinfoleak}
%\end{table}
\begin{table}[http]
	\centering
	\begin{tabular}{ c@{\hspace{0.75\tabcolsep}}|@{\hspace{0.75\tabcolsep}}c @{\hspace{1.5\tabcolsep}}c@{\hspace{1.5\tabcolsep}} c@{\hspace{1.5\tabcolsep}} c@{\hspace{1.5\tabcolsep}} c@{\hspace{1.5\tabcolsep}} c@{\hspace{1.5\tabcolsep}} c}
		&$c_0$ &$\hat{c}_0$ & $c_Z$ &$\hat{c}_Z$& $L$ &$\hat{L}$ &$L-\hat{L}$\\ \hline
		$p_{X,Y}$ & $0.693$ &-& $0.627$ &-&$0.066$&-&-\\  
		 $\hat{p}^1_{X,Y}$ &-&$0.693$& -& $0.844$  &-&$-0.151$ &$0.217$\\
		$\hat{p}^2_{X,Y}$ &-&$0.916$&-& $0.732$ &-&$0.184$ &$-0.118$
	\end{tabular}\caption{Costs and information leakage for different adversaries}\label{tb:costsinfoleak}
\end{table}

From Table~\ref{tb:costsinfoleak}, it is clear that a biased prior always leads to costs no smaller than the omniscience case, which is consistent with~\eqref{eq:costimperfect_0} and~\eqref{eq:costimperfect_z}. Moreover, the information leakage is no longer necessarily nonnegative when the adversary's information is inaccurate. On the other hand, the comparison of information leakage shows that a limited adversary could have higher or lower information leakage depending on the adopted biased prior.
\end{example}

The comparison of information leakage in Example~\ref{exam:infoleak} may give the impression that a limited adversary might even have better inference performance (resulting in a larger information leakage for the user). However, we emphasize that information leakage is not a fair metric when comparing different adversaries. In fact, the adversary knowing a different correlation $\hat{p}_{X,Y}\neq p_{X,Y}$ could have very high initial costs, which leads to misleadingly low information leakage ($\hat{p}^2_{X,Y}$ in Example~\ref{exam:infoleak}). On the other hand, the posterior costs for a limited adversary are always higher than those for the omniscient ones. Therefore, we will adopt the posterior costs as the criterion when we compare adversaries with different information. Note that maximizing posterior costs for an adversary is still consistent with~\eqref{eq:obj}-\eqref{eq:constraintconv} since the private variable $X$ is fixed in~\eqref{eq:obj}-\eqref{eq:constraintconv}.

\subsection{Solution based on concave-convex procedure}
It has been shown in~\cite[Theorem 1]{FduPC-NF:12} that \eqref{eq:obj}-\eqref{eq:constraintconv} is a convex optimization problem. However, when the adversary has limited information, the information leakage $\hat{L}$ is in general not convex, as hinted by \eqref{eq:leakageimperfect}. Fortunately, since the KL divergence is a convex function~\cite[Theorem 2.7.2]{TMC-JAT:12},  the information leakage $\hat{L}$ in~\eqref{eq:leakageimperfect} is a difference of convex functions. Therefore, Problem~\ref{prob:privacyutilityad} is a difference of convex functions (DC) program with polyhedral constraints, and we solve it via the concave-convex procedure (CCCP) \cite{BKS-GRGL:09,HALeT-TPD:18}.

\begin{thom}[General solution as DC programming]
The objective function in Problem~\ref{prob:privacyutilityad} is a difference of convex functions.
\end{thom}
\begin{proof}
Note that the KL divergence is a convex function~\cite[Theorem 2.7.2]{TMC-JAT:12} in both of its arguments, the convexity of the two terms in~\eqref{eq:leakageimperfect} then follows from the facts that: a) $p_{X,Z}$, $\hat{p}_{X,Z}$ and $\hat{p}_{Z}$ are linear functions of $p_{Z|Y}$; b) $\hat{p}_X$ is constant.
\end{proof}

We present the CCCP that solves Problem~\ref{prob:privacyutilityad} in  Algorithm~\ref{alg:DC}, where the matrix variable $\mathbf{M}\in\mathbb{R}^{|\mathcal{Y}|\times|\mathcal{Z}|}$ denotes the matrix representation of the conditional distribution $p_{Z|Y}$ and each row of $\mathbf{M}$ is a probability vector, $f(\mathbf{M})$ and $g(\mathbf{M})$ denote the KL divergence $\mathcal{KL}(p_{X,Z}||\hat{p}_{X}\hat{p}_{Z})$ and $\mathcal{KL}(p_{X,Z}||\hat{p}_{X,Z})$, respectively, and the constraint set $\mathcal{C}$ in line \ref{line:convexp} consists of constraints \eqref{eq:constraintad} and \eqref{eq:simplexad}. The basic idea of CCCP is as follows: at each step, the second term $\mathcal{KL}(p_{X,Z}||\hat{p}_{X,Z})$ in~\eqref{eq:leakageimperfect} is linearized around the current solution so that the resulting programming becomes convex; then, this convex problem is solved to optimality via some efficient algorithms for convex problems, e.g., the interior-point methods; the process iterates until the solution converges or a maximum number of iterations is reached. We note that since the constraint set in our problem is compact, the limit points of the sequence of solutions $\{\mathbf{M}_k\}$ obtained by Algorithm~\ref{alg:DC} are stationary points of Problem~\ref{prob:privacyutilityad} \cite[Theorem 4]{BKS-GRGL:09}.

In order to implement Algorithm~\ref{alg:DC}, we need to calculate the gradient of $g(\mathbf{M})$ as in line~\ref{line:gradient}. We provide the explicit form of the gradient in the following lemma.

\begin{lema}[Gradient of the KL divergence]
Given joint distributions $p_{X,Y}$ and $\hat{p}_{X,Y}$ that satisfy Assumption~\ref{assump:corpri} and~\ref{assump:supportad}\ref{assump:joint}, respectively. Let $\mathbf{P}\in\mathbb{R}^{|\mathcal{X}|\times|\mathcal{Y}|}$ and $\hat{\mathbf{P}}\in\mathbb{R}^{|\mathcal{X}|\times|\mathcal{Y}|}$ be the matrix representations of $p_{X,Y}$ and $\hat{p}_{X,Y}$, respectively. Then, the gradient of $\mathcal{KL}(p_{X,Z}||\hat{p}_{X,Z})$ with respect to the privacy mapping $\mathbf{M}\in\mathbb{R}^{|\mathcal{Y}|\times|\mathcal{Z}|}$ is given by
\begin{equation}\label{eq:gradient}
\frac{\partial \mathcal{KL}(p_{X,Z}||\hat{p}_{X,Z})}{\partial \mathbf{M}}=\mathbf{P}^\top(\log\mathbf{W}+\mathbf{1}_{|\mathcal{X}|\times|\mathcal{Z}|})-\hat{\mathbf{P}}^\top\mathbf{W},
\end{equation}
where the logarithm is component-wise and 
\begin{equation*}
\mathbf{W}=(\mathbf{P}\mathbf{M})\oslash(\hat{\mathbf{P}}\mathbf{M}).
\end{equation*}
\end{lema}
\begin{proof}
We expand the KL divergence $\mathcal{KL}(p_{X,Z}||\hat{p}_{X,Z})$ as 
\begin{equation*}
	\mathcal{KL}(p_{X,Z}||\hat{p}_{X,Z})=\sum_{x,z}\mathbf{P}_{x,*}^\top\mathbf{M}_{*,z}\log\frac{\mathbf{P}_{x,*}^\top\mathbf{M}_{*,z}}{\hat{\mathbf{P}}_{x,*}^\top\mathbf{M}_{*,z}}.
\end{equation*}
Fix an observation $z$, we have that 
\begin{align}\label{eq:gradientcolumn}
	\begin{split}
	&\quad\frac{\partial \mathcal{KL}(p_{X,Z}||\hat{p}_{X,Z})}{\partial \mathbf{M}_{*,z}}\\
&=\sum_{x\in\mathcal{X}}\big(\mathbf{P}_{x,*}\cdot(\log(\frac{\mathbf{P}_{x,*}^\top\mathbf{M}_{*,z}}{\hat{\mathbf{P}}_{x,*}^\top\mathbf{M}_{*,z}})+1)-\hat{\mathbf{P}}_{x,*}\cdot\frac{\mathbf{P}_{x,*}^\top\mathbf{M}_{*,z}}{\hat{\mathbf{P}}_{x,*}^\top\mathbf{M}_{*,z}}\big)\\
&=\sum_{x\in\mathcal{X}}\big(\mathbf{P}_{x,*}\cdot(\log(\mathbf{W}_{x,z})+1)-\hat{\mathbf{P}}_{x,*}\cdot\mathbf{W}_{x,z}\big)\\
&=\mathbf{P}^\top(\log(\mathbf{W}_{*,z})+\mathbf{1}_{|\mathcal{X}|})-\hat{\mathbf{P}}^\top\mathbf{W}_{*,z}\big).
	\end{split}
\end{align}
Then, we obtain \eqref{eq:gradient} by concatenating the gradients with respect to each column of  $\mathbf{M}$ in~\eqref{eq:gradientcolumn}.

\end{proof}

\begin{algorithm}
	\caption{DC programming for Problem~\ref{prob:privacyutilityad}}\label{alg:DC}
	\begin{algorithmic}[1]
		\renewcommand{\algorithmicrequire}{\textbf{Input:}}
		\renewcommand{\algorithmicensure}{\textbf{Output:}}
		\REQUIRE Joint distributions $p_{X,Y}(x,y)$ and $\hat{p}_{X,Y}(x,y)$ that satisfy Assumption~\ref{assump:corpri} and~\ref{assump:supportad}\ref{assump:joint}, respectively
		\renewcommand{\algorithmicrequire}{\textbf{Parameters:}}
		\REQUIRE the error tolerance $\epsilon>0$, the maximum number of iterations \texttt{MaxIter}\\				
		\ENSURE The optimal privacy mapping in Problem~\ref{prob:privacyutilityad}
		\renewcommand{\algorithmicrequire}{\textbf{Initialize:}}
		\REQUIRE $\mathbf{M}_0\in\mathbb{R}^{|\mathcal{Y}|\times|\mathcal{Z}|}$  satisfying~\eqref{eq:constraintad}-\eqref{eq:simplexad}, $k\gets0$\\
		\WHILE{\texttt{TRUE}}
			\STATE\label{line:gradient} Compute $\mathbf{G}_k\gets\nabla g(\mathbf{M}_k)$ via~\eqref{eq:gradient}
%			\frac{\partial g(\mathbf{x})}{\partial\mathbf{x}}|_{\mathbf{x}=\mathbf{x}_k}
			\STATE\label{line:convexp} Compute $\mathbf{M}_{k+1}\gets\argmin_{\mathbf{M}\in\mathcal{C}}\{f(\mathbf{M})-\mathbf{1}_{|\mathcal{Y}|}^\top(\mathbf{G}_k\odot\mathbf{M})\mathbf{1}_{\mathcal{Z}}\}$
			\IF{$\|\mathbf{M}_{k+1}-\mathbf{M}_{k}\|_F\leq\epsilon$ or $k>\texttt{MaxIter}$}
		 	\RETURN $\mathbf{M}_{k+1}$
			\ENDIF
			\STATE $k\gets k+1$
			\ENDWHILE
	\end{algorithmic}
\end{algorithm}

\subsection{Sufficient conditions for convexity of the objective function}
When an adversary has perfect information, i.e., $\hat{p}_{X,Y}=p_{X,Y}$,  Problem~\ref{prob:privacyutilityad} is convex. In this subsection, we derive a sufficient condition on $\hat{p}_{X,Y}$ so that Problem~\ref{prob:privacyutilityad} remains convex. Then, we can apply interior-point methods to find a globally optimal solution when the sufficient condition is satisfied. 

\begin{thom}[Condition for convexity of the information leakage minimization]\label{thm:convexity}
	Given joint distributions $p_{X,Y}$ and $\hat{p}_{X,Y}$ that satisfy Assumption~\ref{assump:corpri} and~\ref{assump:supportad}\ref{assump:joint}, respectively. Let $\mathbf{P}\in\mathbb{R}^{|\mathcal{X}|\times|\mathcal{Y}|}$ and $\hat{\mathbf{P}}\in\mathbb{R}^{|\mathcal{X}|\times|\mathcal{Y}|}$ be the matrix representations of $p_{X,Y}$ and $\hat{p}_{X,Y}$, respectively. If
	for all $x\in\mathcal{X}$,
	 \begin{align}\label{eq:convexitycondition}
	 	\begin{split}
	(2\mathbf{P}_{x,*}^\top\hat{\mathbf{P}}_{x,*}-r_x^{\textup{L}}\hat{\mathbf{P}}_{x,*}^\top\hat{\mathbf{P}}_{x,*})\cdot(r_x^{\textup{L}}\hat{\mathbf{P}}_Y^\top\hat{\mathbf{P}}_Y)\geq (\mathbf{P}_{x,*}^\top\hat{\mathbf{P}}_Y)^2,\\
	(2\mathbf{P}_{x,*}^\top\hat{\mathbf{P}}_{x,*}-r_x^{\textup{U}}\hat{\mathbf{P}}_{x,*}^\top\hat{\mathbf{P}}_{x,*})\cdot(r_x^{\textup{U}}\hat{\mathbf{P}}_Y^\top\hat{\mathbf{P}}_Y)\geq (\mathbf{P}_{x,*}^\top\hat{\mathbf{P}}_Y)^2,
	 	\end{split}	
\end{align}
where $\hat{\mathbf{P}}_Y$ is the vector representing the marginal distribution of $\hat{p}_{X,Y}$ with respect to $Y$ and \begin{equation}\label{eq:ratio}
		r_{x}^{\textup{L}}=\min_{y\in\mathcal{Y}:\mathbf{P}_{x,y}>0}\frac{\mathbf{P}_{x,y}}{\hat{\mathbf{P}}_{x,y}}\quad\textup{and}\quad		r_{x}^{\textup{U}}=\max_{y\in\mathcal{Y}:\mathbf{P}_{x,y}>0}\frac{\mathbf{P}_{x,y}}{\hat{\mathbf{P}}_{x,y}},
	\end{equation}
	then Problem~\ref{prob:privacyutilityad} is a convex optimization problem.
\end{thom}
\begin{proof}
We postpone the proof to Appendix~\ref{appendix:convexity}.
\end{proof}

\begin{remark}[Convexity of \eqref{eq:obj}-\eqref{eq:constraintconv}]
When $\hat{p}_{X,Y}=p_{X,Y}$, we have $r_{x}^{\textup{L}}=r_{x}^{\textup{U}}=1$ for all $x\in\mathcal{X}$ in Theorem~\ref{thm:convexity}, and \eqref{eq:convexitycondition} follows automatically from the Cauchy–Schwarz inequality. Therefore, problem~\eqref{eq:obj}-\eqref{eq:constraintconv} is convex.
\end{remark}

\section{A limited adversary with unknown correlations}\label{sec:expected}
The results in Section~\ref{sec:known} apply to the case when the adversary's information $\hat{p}_{X,Y}$ is known to the user. For example, it may be known what public data is available to the adversary to build the correlation $\hat{p}_{X,Y}$. In this section, we treat the case when $\hat{p}_{X,Y}$ used by the adversary is not known. Instead, we assume a distribution for $\hat{p}_{X,Y}$ and maximize the corresponding average posterior costs. We can interpret this approach as finding a mapping that protects the private information against a family of adversaries, and each has a different level of information. Although we can apply the methods developed in this section to any distribution for $\hat{p}_{X,Y}$ over the probability simplex, we assume, as a concrete example, that $\hat{p}_{X,Y}$ follows a Dirichlet distribution \cite[Chap. 49]{SK-NB-NLJ:00}.

\subsection{Preliminaries of Dirichlet distribution}
Since $\hat{p}_{X,Y}$ is a probability distribution, one choice for the distribution of $\hat{p}_{X,Y}$ is the Dirichlet distribution. In this subsection, we review some basics of the Dirichlet distribution. A Dirichlet distribution $\mathcal{D}_N(\bm{\alpha})$ in dimension $N$ with parameters $\bm{\alpha}=\begin{bmatrix}
\alpha_1&\dots&\alpha_N
\end{bmatrix}^\top$ is a continuous probability distribution over the probability simplex $\Delta_N$. The probability density function of a Dirichlet distribution $\mathcal{D}_N(\bm{\alpha})$ is given by
\begin{equation*}
p_{\mathcal{D}}(x_1,\dots,x_N;\alpha_1,\dots,\alpha_N) =\frac{1}{\mathcal{B}(\bm{\alpha})}\Pi_{i=1}^Nx_i^{\alpha_i-1},
\end{equation*}
where
\begin{equation*}
\mathcal{B}(\bm{\alpha})=\frac{\Pi_{i=1}^N\Gamma(\alpha_i)}{\Gamma(\alpha_0)},
\end{equation*}
$\alpha_0=\sum_{i=1}^N\alpha_i$, and $\Gamma(\cdot)$ is the gamma function. For a Dirichlet random variable $X\sim\mathcal{D}_N(\bm{\alpha})$, the expectations $\mathbb{E}[X]$ and $\mathbb{E}[\log X]$ and the variance $\textup{Var}[X]$ are given by
\begin{align}\label{eq:dirichlet}
	\begin{split}
&\mathbb{E}[X]=\frac{1}{\alpha_0}\begin{bmatrix}
	\alpha_1&\cdots&\alpha_N\end{bmatrix}^\top,\\
&\mathbb{E}[\log X] = \begin{bmatrix}
	\psi(\alpha_1)&\cdots&\psi(\alpha_N)\end{bmatrix}^\top-\psi(\alpha_0)\mathbf{1}_N,\\
&\textup{Var}[X]=\frac{1}{\alpha_0+1}\begin{bmatrix}
	\frac{\alpha_1}{\alpha_0}(1-\frac{\alpha_1}{\alpha_0})&\cdots&\frac{\alpha_N}{\alpha_0}(1-\frac{\alpha_N}{\alpha_0})
\end{bmatrix}^\top,
	\end{split}
\end{align}
where $\psi(x)=\frac{d\log\Gamma(x)}{dx}$ is the digamma function. Dirichlet distributions are closed under aggregation, i.e., if $X\sim\mathcal{D}_N(\bm{\alpha})$, then we have
\begin{equation*}
(X_1+X_2,X_3,\cdots,X_N)\sim\mathcal{D}_N(\alpha_1+\alpha_2,\alpha_3,\cdots,\alpha_N).
\end{equation*}

\subsection{Expected posterior costs and problem of interest}
When we translate the limited adversary with the prior information $\hat{p}_{X,Y}$ following a probability distribution to a family of adversaries with different $\hat{p}_{X,Y}$'s, a more appropriate performance metric is the posterior cost $\hat{c}_Z$ in~\eqref{eq:costimperfect_z} as discussed in Example~\ref{exam:infoleak}. We expand $\hat{c}_Z$ as follows,
\begin{align}\label{eq:postcost}
	\begin{split}
	\hat{c}_Z&=H(X|Z)+\mathcal{KL}(p_{X|Z}||\hat{p}_{X|Z})\\
&=-\sum_{x,z}p_{X,Z}(x,z)\log \hat{p}_{X|Z}(x|z)\\
&=\sum_{x,z}p_{X,Z}(x,z)\log \frac{\hat{p}_{Z}(z)}{\hat{p}_{X,Z}(x,z)}.
	\end{split}
\end{align}
Suppose $\hat{p}_{X,Y}$ follows a Dirichlet distribution $\mathcal{D}_{|\mathcal{X}|\times|\mathcal{Y}|}$ with parameters $\{\alpha_{x,y}\}_{x\in\mathcal{X},y\in\mathcal{Y}}$, then the posterior cost in~\eqref{eq:postcost} becomes a random variable and we design a mapping $p_{Z|Y}$ to maximize the expectation of the posterior cost. 
Following similar arguments in Lemma~\ref{lemma:wellposedness} and Lemma~\ref{lemma:infinitecosts}, we impose the following assumption on the parameters $\{\alpha_{x,y}\}_{x\in\mathcal{X},y\in\mathcal{Y}}$ of the Dirichlet distribution in order for the problem of interest to be well-defined.
\begin{assumption}[Parameters of the Dirichlet distribution]\label{assump:dirichletpara}
Given a joint distribution $p_{X,Y}$ that satisfies Assumption~\ref{assump:corpri}, the parameters $\{\alpha_{x,y}\}_{x\in\mathcal{X},y\in\mathcal{Y}}$ of the Dirichlet distribution $\mathcal{D}_{|\mathcal{X}|\times|\mathcal{Y}|}$ describing $\hat{p}_{X,Y}$ satisfy that for any $x\in\mathcal{X}$ and $y\in\mathcal{Y}$, if $p_{X,Y}(x,y)>0$, then $\alpha_{x,y}>0$.  
\end{assumption}

We study the following problem in the rest of this section.

\begin{problem}[Privacy-utility trade-offs against a limited adversary with unknown correlations]\label{prob:privacyutilitymultiple}
Given a joint distribution $p_{X,Y}$ that satisfies Assumption~\ref{assump:corpri} and a Dirichlet distribution $\mathcal{D}_{|\mathcal{X}|\times|\mathcal{Y}|}$ with parameters $\{\alpha_{x,y}\}_{x\in\mathcal{X},y\in\mathcal{Y}}$ satisfying Assumption~\ref{assump:dirichletpara}, find a privacy mapping $p_{Z|Y}$ such that the expected posterior cost $\mathbb{E}[\hat{c}_Z]$ is maximized under a utility loss constraint, i.e., solve the following optimization problem
	\begin{subequations}\nonumber
		\begin{align}
			\maximize_{p_{Z|Y}} \quad & \mathbb{E}[\hat{c}_Z]\\
			\st \quad & \E_{Y,Z}[d(y,z)] \leq\delta,\\
			&p_{Z|Y}(\cdot|y)\in\Delta_{|\mathcal{Z}|},\quad\forall y\in\mathcal{Y}.
		\end{align}
	\end{subequations}
\end{problem}

Since an analytic expression for the objective function in Problem~\ref{prob:privacyutilitymultiple} is not available, we propose to approximate the objective function by its lower bound.

\subsection{Lower bounding the expected posterior cost}
In this subsection, we provide a lower bound for $\mathbb{E}[\hat{c}_Z]$ and propose to maximize the lower bound instead of the objective function in Problem~\ref{prob:privacyutilitymultiple} directly.

\begin{thom}[Lower bound for expected posterior cost]
Given a joint distribution $p_{X,Y}$ that satisfies Assumption~\ref{assump:corpri} and a Dirichlet distribution $\mathcal{D}_{|\mathcal{X}|\times|\mathcal{Y}|}$ with parameters $\{\alpha_{x,y}\}_{x\in\mathcal{X},y\in\mathcal{Y}}$ satisfying Assumption~\ref{assump:dirichletpara}, then the following inequality holds
\begin{multline}\label{eq:lowerbound}
	\mathbb{E}[\hat{c}_Z]\geq\sum_{x,z}p_{X,Z}(x,z)\big(\log(\sum_{y\in\mathcal{Y}}e^{\psi(\alpha_y)-\psi(\alpha_0)}p_{Z|Y}(z|y))\\
	-\log(\sum_{y\in\mathcal{Y}}\frac{\alpha_{x,y}}{\alpha_0}p_{Z|Y}(z|y))\big),
\end{multline}
where $\alpha_y=\sum_{x\in\mathcal{X}}\alpha_{x,y}$ and $\alpha_0=\sum_{y\in\mathcal{Y}}\alpha_{y}$.
\end{thom}
\begin{proof}
Since $p_{X,Z}(x,z)$'s are deterministic in~\eqref{eq:postcost}, we fix a pair of $x\in\mathcal{X}$ and $z\in\mathcal{Z}$ and compute
\begin{align}\label{eq:lowerboundderivation}
	\begin{split}
&\quad\mathbb{E}[\log \frac{\hat{p}_{Z}(z)}{\hat{p}_{X,Z}(x,z)}]\\
&=\mathbb{E}[\log \hat{p}_{Z}(z)]-\mathbb{E}[\log \hat{p}_{X,Z}(x,z)]\\
&\geq\mathbb{E}[\log \hat{p}_{Z}(z)]-\log(\mathbb{E}[ \hat{p}_{X,Z}(x,z)])\\
&=\mathbb{E}[\log \big(\sum_{y\in\mathcal{Y}}e^{\log(\hat{p}_Y(y)p_{Z|Y}(z|y))}]\big)-\log(\mathbb{E}[ \hat{p}_{X,Z}(x,z)])\\
&\geq\log\big(\sum_{y\in\mathcal{Y}}e^{\mathbb{E}[\log(\hat{p}_Y(y)p_{Z|Y}(z|y))]}\big)-\log(\mathbb{E}[ \hat{p}_{X,Z}(x,z)])\\
&=\log\big(\sum_{y\in\mathcal{Y}}e^{\psi(\alpha_y)-\psi(\alpha_0)}p_{Z|Y}(z|y))\big)\\
&\quad-\log(\sum_{y\in\mathcal{Y}}\frac{\alpha_{x,y}}{\alpha_0}p_{Z|Y}(z|y))\\
&=\log\big(\sum_{y\in\mathcal{Y}}e^{\psi(\alpha_y)-\psi(\alpha_0)}p_{Z|Y}(z|y))\big)\\
&\quad-\log(\sum_{y\in\mathcal{Y}}\frac{\alpha_{x,y}}{\alpha_0}p_{Z|Y}(z|y)),
	\end{split}
\end{align}
where the inequalities follow from the Jensen's inequality and the facts that $\log(\cdot)$ is concave and the log-sum-exp function is convex. Then, we obtain \eqref{eq:lowerbound} by summing~\eqref{eq:lowerboundderivation} over $x$ and $z$ with weights $p_{X,Z}(x,z)$.
\end{proof}

\begin{remark}[Consistency with the case of known adversary's prior]
We note that, when the variance of the Dirichlet distribution $\mathcal{D}_{|\mathcal{X}|\times\mathcal{Y}}$ goes to zero and the expectation stays unchanged, i.e., $\alpha_0\rightarrow\infty$ and $\frac{\alpha_{x,y}}{\alpha_0}$ remains constant, the bound in~\eqref{eq:lowerbound} recovers the posterior cost~\eqref{eq:postcost} as if there were a limited adversary with the prior $\hat{p}_{X,Y}(x,y)=\frac{\alpha_{x,y}}{\alpha_0}$. Specifically, when the variance goes to zeros, we have $\hat{p}_Y(y)=\frac{\alpha_y}{\alpha_0}$ almost surely and $\mathbb{E}[\log(\hat{p}_Y(y))]=\log\frac{\alpha_y}{\alpha_0}$, and the fourth line of \eqref{eq:lowerboundderivation} becomes the same as~\eqref{eq:postcost}. 
\end{remark}

\begin{remark}[Bounds for other distributions]
Other than the fact that $\mathbb{E}[\log X]$ has a closed-form expression for a Dirichlet random variable $X$, we did not use any other properties of the Dirichlet distribution in deriving the lower bound~\eqref{eq:lowerbound} (the aggregation property of the Dirichlet distribution is not necessary for the derivation). Therefore, \eqref{eq:lowerbound} is applicable to any other distributions for $\hat{p}_{X,Y}$.
\end{remark}

We denote the lower bound for $\mathbb{E}[\hat{c}_Z]$ in~\eqref{eq:lowerbound} by $\underline{\hat{c}}_Z$ and turn to solve the following optimization problem.

\begin{problem}[Approximate privacy-utility trade-offs against a limited adversary with unknown correlations]\label{prob:privacyutilitymultiplelower}
	Given a joint distribution $p_{X,Y}$ that satisfies Assumption~\ref{assump:corpri} and a Dirichlet distribution $\mathcal{D}_{|\mathcal{X}|\times|\mathcal{Y}|}$ with parameters $\{\alpha_{x,y}\}_{x\in\mathcal{X},y\in\mathcal{Y}}$ satisfying Assumption~\ref{assump:dirichletpara}, find a privacy mapping $p_{Z|Y}$ such that the lower bound $\underline{\hat{c}}_Z$ for $\mathbb{E}[\hat{c}_Z]$ given in~\eqref{eq:lowerbound} is maximized under a utility loss constraint, i.e., solve the following optimization problem
	\begin{subequations}\nonumber
		\begin{align}
			\maximize_{p_{Z|Y}} \quad & \underline{\hat{c}}_Z\\
			\st \quad & \E_{Y,Z}[d(y,z)] \leq\delta,\\
			&p_{Z|Y}(\cdot|y)\in\Delta_{|\mathcal{Z}|},\quad\forall y\in\mathcal{Y}.
		\end{align}
	\end{subequations}
\end{problem}

By adding and subtracting terms in $\underline{\hat{c}}_Z$, we can again reorganize the objective function in Problem~\ref{prob:privacyutilitymultiplelower} into the form of a DC. Specifically,
\begin{align}\label{eq:lowerboundconvex}
\begin{split}
&\quad\underline{\hat{c}}_Z\\
&=	\sum_{x,z}p_{X,Z}(x,z)\big(\log(\frac{\sum_{y\in\mathcal{Y}}e^{\psi(\alpha_y)-\psi(\alpha_0)}p_{Z|Y}(z|y)}{p_{X,Z}(x,z)})\\
&\quad+\log\frac{p_{X,Z}(x,z)}{\sum_{y\in\mathcal{Y}}\frac{\alpha_{x,y}}{\alpha_0}p_{Z|Y}(z|y)}\big)\\
&=	-\underbrace{\sum_{x,z}p_{X,Z}(x,z)\log(\frac{p_{X,Z}(x,z)}{\sum_{y\in\mathcal{Y}}e^{\psi(\alpha_y)-\psi(\alpha_0)}p_{Z|Y}(z|y)})}_{\textup{first term}}\\
&\quad+\underbrace{\sum_{x,z}p_{X,Z}(x,z)\log\frac{p_{X,Z}(x,z)}{\sum_{y\in\mathcal{Y}}\frac{\alpha_{x,y}}{\alpha_0}p_{Z|Y}(z|y)}}_{\textup{second term}},
\end{split}
\end{align}
where the convexity of both terms in~\eqref{eq:lowerboundconvex} as functions of $p_{Z|Y}$ follows from the facts that  $t\log t$ is convex and the perspective function of a convex function is convex \cite[Chap. 3.2.6]{SB-LV:04}. Note that the lower bound $\underline{\hat{c}}_Z$ of $\mathbb{E}[\hat{c}_Z]$ in~\eqref{eq:lowerboundconvex} has a very similar form as the posterior cost $\hat{c}_Z$ in~\eqref{eq:postcost}.  Finally, we can solve Problem~\ref{prob:privacyutilitymultiplelower} by Algorithm~\ref{alg:DC} after rewriting the objective function as minimizing  $-\underline{\hat{c}}_Z$.

\section{Numerical examples}\label{sec:numerics}
This section provides a numerical example using a census data set \cite{DD-CG:2017,SS-AZ-FduPC-SB-NF-BK-PO-NT:15} to illustrate the presented results. 

\subsection{Simulation setup}
The data set contains personal information of $48842$ individuals, and each individual has $14$ recorded attributes, e.g., age, gender, education, income level, race, marriage status. In our example, we take the ``income level'' as the private information $X$ with support  $\mathcal{X}=\{\texttt{high},\texttt{low}\}$, where \texttt{high} and \texttt{low} correspond to income levels ``$\geq50K$'' and ``$<50K$'', respectively. We select the tuple of attributes $(``\textup{age}",``\textup{gender}",``\textup{education}")$ to be the useful information $Y$, and each attribute has the following possibilities
\begin{enumerate}
\item $\text{``age''}\in\{\texttt{young},\texttt{middle-aged},\texttt{senior}\}$ where we assign \texttt{young}, \texttt{middle-aged} and \texttt{senior} to people whose age falls in the ranges $[0,30]$, $(30,60]$, and $(60,100]$, respectively;
\item $\text{``gender''}\in\{\texttt{male},\texttt{female}\}$;
\item $\text{``education''}\in\{\texttt{others},\texttt{college},\texttt{graduate}\}$ where \texttt{others} is for people who are high-school graduates or under, \texttt{college} is for people having bachelors or equivalent degrees (e.g., professional school, some college experience), and \texttt{graduate} is for people  having graduate degrees (masters or doctorates).
\end{enumerate}
The support $\mathcal{Y}$ of the public information $Y$ consists of $18$ possible combinations of the attribute tuple $(``\textup{age}",``\textup{gender}",``\textup{education}")$. Let $n_{x,y}$ be the number of people with attributes $x\in\mathcal{X}$ and $y\in\mathcal{Y}$, then we set the joint probability $p_{X,Y}(x,y)$ to be $\frac{n_{x,y}}{48842}$. For a given $y\in\mathcal{Y}$, we generate the released information $z$ by erasing zero, one or two  dimensions of $y$, and the distortion $d(y,z)$ is the number of erased features. As a concrete example, let $y=(\texttt{young},\texttt{male},\texttt{college})$, then the released information $z$ has the following possibilities: (\texttt{young}, \texttt{male}, \texttt{college}), ($-$, \texttt{male}, \texttt{college}), (\texttt{young}, $-$, \texttt{college}), (\texttt{young}, \texttt{male}, $-$),(\texttt{young}, $-$, $-$), ($-$, \texttt{male}, $-$), ($-$, $-$, \texttt{college}), where $-$ represents an erasure. The support $\mathcal{Z}$ of $Z$ is of size $47$ by construction.

We implement Algorithm~\ref{alg:DC} to solve Problem~\ref{prob:privacyutilityad} and Problem~\ref{prob:privacyutilitymultiplelower} in the simulation, and the parameters $\epsilon$ and \texttt{MaxIter} are $10^{-6}$ and $100$, respectively. We run the algorithm from $10$ randomly sampled initial conditions and adopt the best found solution. We solve the inner convex problem in line~\ref{line:convexp} of Algorithm~\ref{alg:DC} via CVX in MATLAB \cite{cvx, MG-SB:08}. 

\subsection{Limited adversaries with biased correlations}
In this subsection, we solve Problem~\ref{prob:privacyutilityad} for a few randomly generated adversaries and show the improved privacy-utility trade-offs. We obtain the joint distributions $\hat{p}_{X,Y}$'s by perturbing the elements of the joint distribution $p_{X,Y}$ by a certain percentage, followed by normalization. Specifically, given a percentage level $\gamma\in(0,1)$,  each element of $p_{X,Y}$ is first multiplied by a uniform random variable over $[1-\gamma,1+\gamma]$, then $\hat{p}_{X,Y}$ is constructed by normalizing the perturbed $p_{X,Y}$. In our simulation, we choose the percentage levels to be $\gamma\in\{10\%,25\%,50\%\}$ and show the distances between $p_{X,Y}$ and the generated $\hat{p}_{X,Y}$'s measured by the Frobenius norm in the legend of Fig.~\ref{fig:single}.

\begin{figure}[http]
	\centering
	\subfigure[Information leakge]{\includegraphics[scale=0.45]{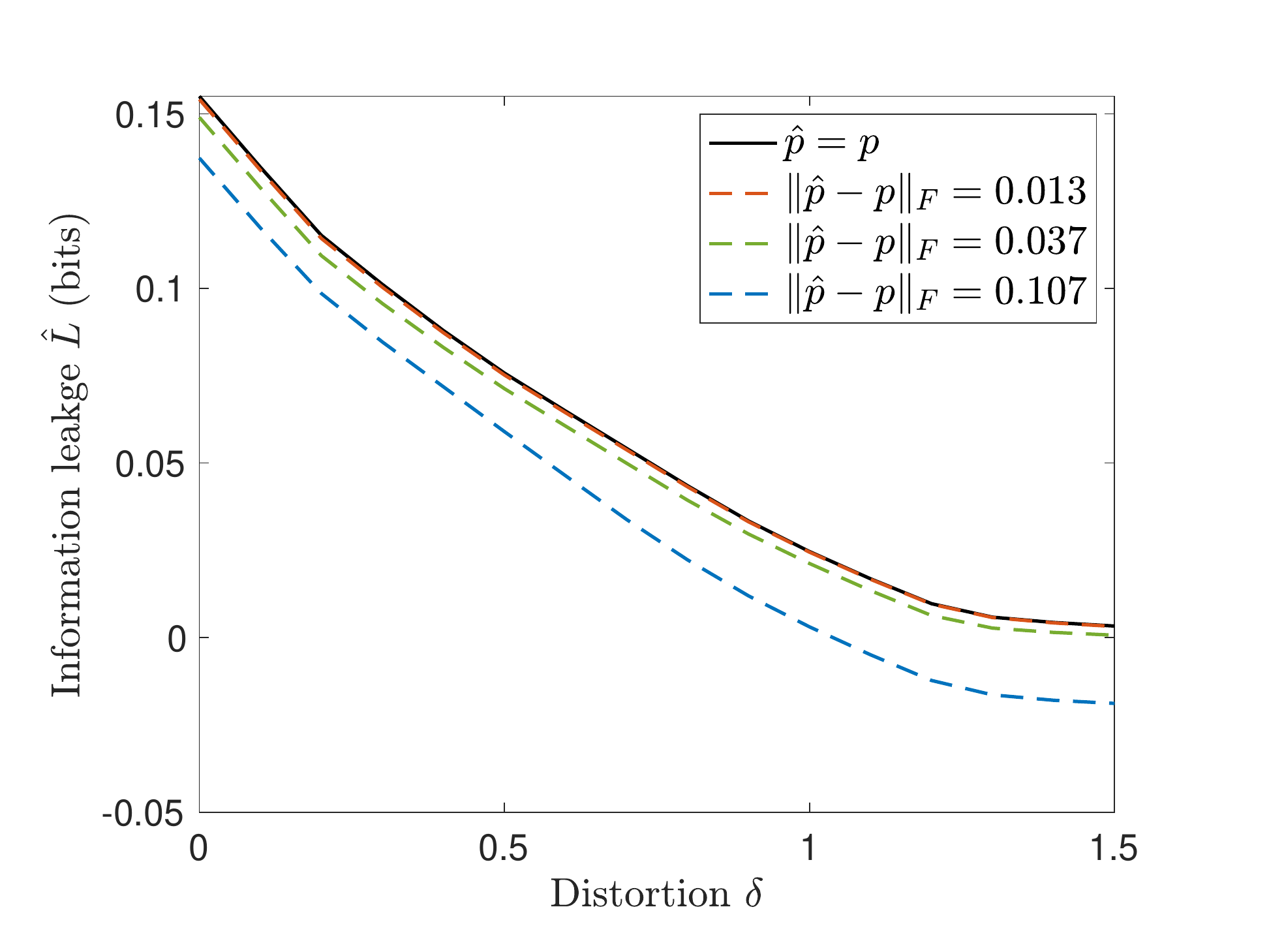}\label{fig:infoleak_single}}
	\subfigure[Posterior costs]{\includegraphics[scale=0.45]{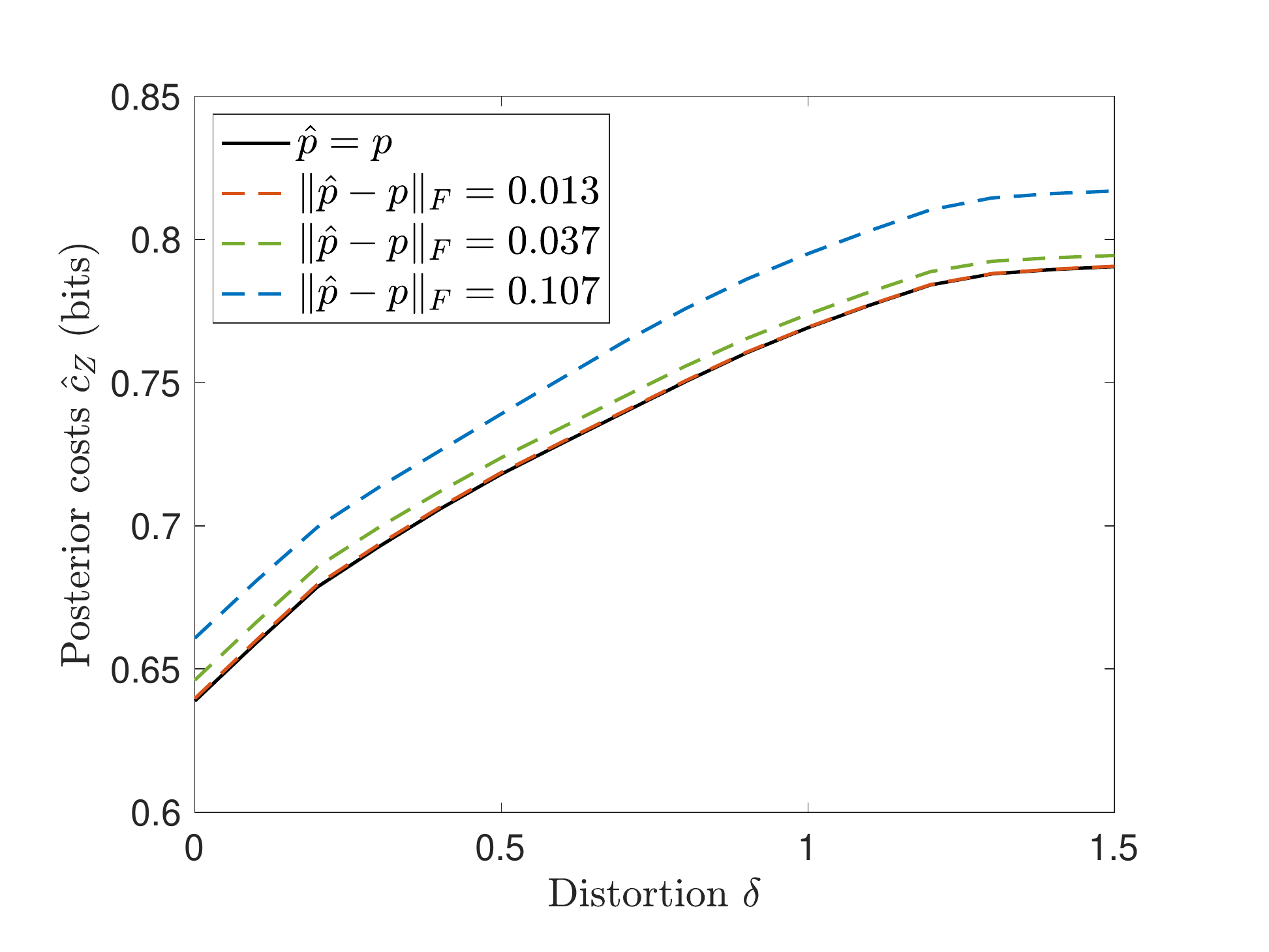}\label{fig:postcost_single}}
	\caption{Information leakage and posterior costs against limited adversaries.}\label{fig:single}
\end{figure}

In Fig.~\ref{fig:infoleak_single} and~\ref{fig:postcost_single}, the x-axis is the distortion level ranging from $0$ to $1.5$ with an increment of $0.1$, and the y-axis shows the information leakage and the posterior costs (evaluated in bits with the logarithm to the base $2$), respectively. We observe that the difference between $\hat{p}_{X,Y}$ and $p_{X,Y}$ leads to lower information leakage for the user and correspondingly higher costs for the adversaries. Under the same distortion level, a larger difference leads to worse inference performance for the adversaries quantified by the posterior costs. Moreover, when the distance between $\hat{p}_{X,Y}$ and $p_{X,Y}$ is relatively high, the information leakage can be negative as demonstrated by the blue dashed line in Fig.~\ref{fig:infoleak_single} in the high-distortion regime.

\subsection{Limited adversaries with unknown correlations}
In this subsection, we solve Problem~\ref{prob:privacyutilitymultiplelower} for cases when the adversary's information about the correlation is modeled by a Dirichlet distribution $\mathcal{D}_{|\mathcal{X}|\times|\mathcal{Y}|}(\bm{\alpha})$. Let a scale parameter be $\nu\in\{0.008,0.01,0.05\}$. Then for $x\in\mathcal{X}$ and $y\in\mathcal{Y}$, ${\alpha}_{x,y}$ is set to be $\nu n_{x,y}$. Note that a smaller $\nu$ implies a smaller $\alpha_0$ and by~\eqref{eq:dirichlet}, a higher variance. On the other hand, the expectations of these Dirichlet distributions are equal to $p_{X,Y}$. We can interpret the distribution's variance as the user's confidence, i.e., a lower variance corresponds to a conservative user. We evaluate the performance of the solutions to Problem~\ref{prob:privacyutilitymultiplelower} against these three different Dirichlet distributions through empirical averages, i.e., we sample $100$ $\hat{p}_{X,Y}$'s from the respective Dirichlet distribution and calculate the average empirical posterior costs. We also include the performance of the nominal privacy mapping obtained by assuming $\hat{p}_{X,Y}=p_{X,Y}$. Note that since the true $\hat{p}_{X,Y}$ follows different Dirichlet distributions in different cases, the performance of the nominal privacy mapping also differs. The results are reported in Fig.~\ref{fig:expected}. In the legend of the figure, H, M, and L represent the cases of high ($\nu=0.008$), medium ($\nu=0.01$) and low ($\nu=0.05$) variances, respectively. Comparing solid lines with different colors, we observe that higher variance leads to higher posterior costs, which results from the fact that when the variance is high, many $\hat{p}_{X,Y}$'s are away from ${p}_{X,Y}$. The difference between solid and dashed lines with the same color shows that the solution to Problem~\ref{prob:privacyutilitymultiplelower} is superior to the nominal solution in the sense that it causes higher posterior costs for the adversaries. 

\begin{figure}[http]
	\centering
\includegraphics[scale=0.45]{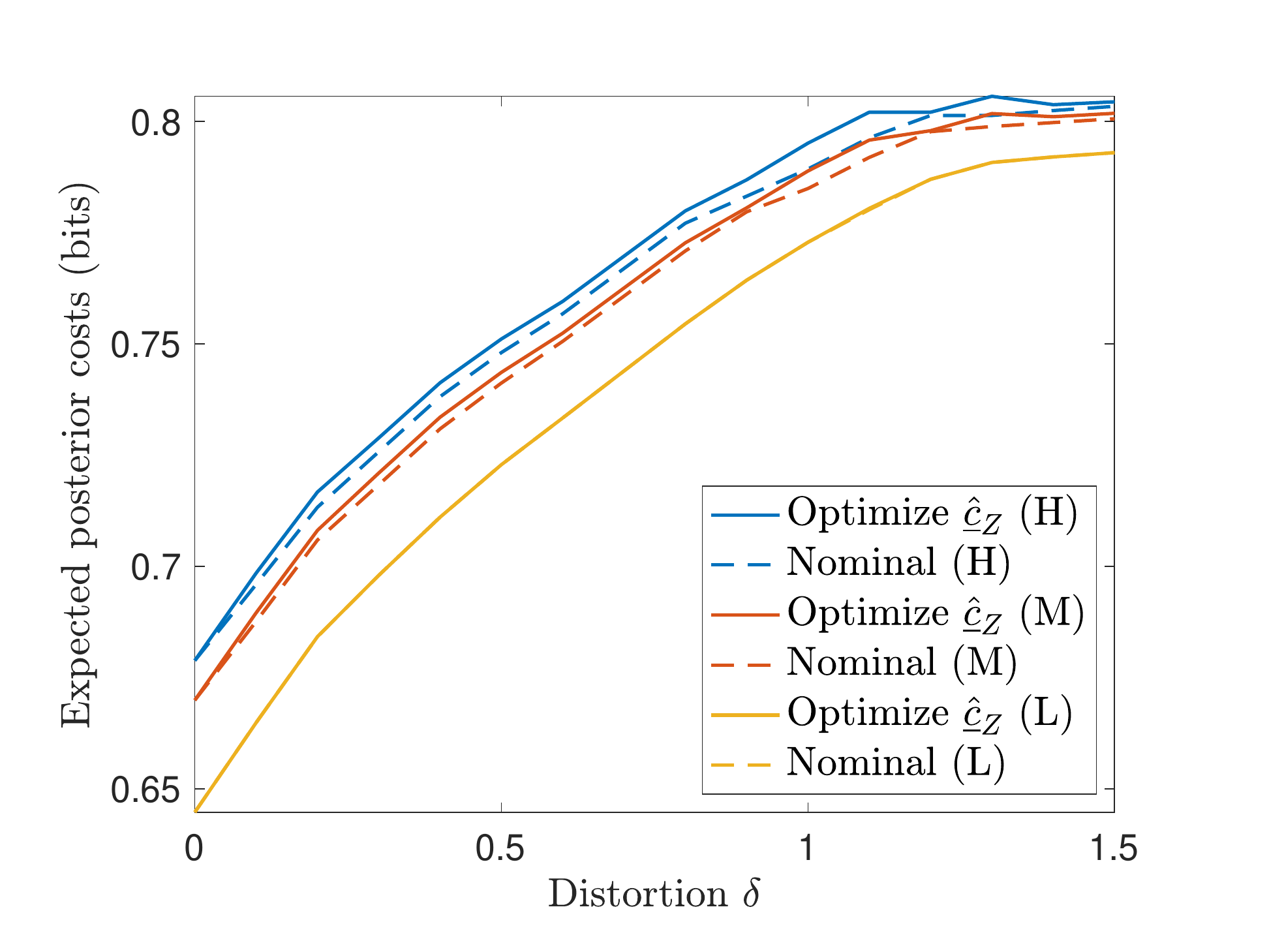}
	\caption{Expected posterior costs against a limited adversary whose information about the correlation is described by a Dirichlet distribution. H, M, L in the figure legend stand for cases of Dirichlet distributions with high, medium and low variances, respectively. The ``nominal'' represents the performance of the solution obtained by assuming $\hat{p}_{X,Y}=p_{X,Y}$.}\label{fig:expected}
\end{figure}

\section{Conclusion}\label{sec:conclusion}
We studied privacy-utility trade-offs against a limited adversary with biased statistical information regarding the underlying correlated private and useful information. We identified minimal assumptions on the adversary's information so that the privacy metrics and the design of the probabilistic privacy mapping are well-defined. We further formulated the design problem as a DC program and solved it via CCCP. When the adversary's information is not precisely available, we adopted a Bayesian view and sought to optimize the average posterior costs for the adversary. 

We exemplified the impact of the information asymmetry between the user and the adversary using mutual information as the underlying privacy metric. For future work, we will further systematically investigate similar impacts on other privacy metrics. On the other hand, it is also interesting to study how one could exploit such impacts to effectively convey information to the intended receivers and hide information from the adversaries, given that these two receivers have different information availability.

\appendices

\section{Proof of Theorem~\ref{thm:convexity}}\label{appendix:convexity}
\begin{proof}
	Note that in Problem~\ref{prob:privacyutilityad}, the constraints~\eqref{eq:constraintad} and \eqref{eq:simplexad} are linear equality and inequality constraints and thus convex. Therefore,  Problem~\ref{prob:privacyutilityad} is convex if the objective function $\hat{L}$ is convex in the variable $\mathbf{M}\in\mathbb{R}^{|\mathcal{Y}\times|\mathcal{Z}|}$ in the domain of the function. We prove the convexity of $\hat{L}$ under condition~\eqref{eq:convexitycondition} by showing that $\hat{L}$ can be represented as a sum of convex functions of vector variables and each convex function has a positive semidefinite Hessian matrix in its domain.
	
	From~\eqref{eq:leakgead}, we have
	\begin{align*}
		\hat{L}&=\sum_{x,z}p_{X,Z}(x,z)\log \frac{\hat{p}_{Z|X}(z|x)}{ \hat{p}_Z(z)}\\
		&=\sum_{x,z}(\mathbf{P}_{x,*}^\top\mathbf{M}_{*,z})\log \frac{\hat{\mathbf{P}}_{x,*}^\top\mathbf{M}_{*,z}}{ (\mathbf{P}_{x,*}^\top\mathbf{1}_{|\mathcal{Y}|})(\hat{\mathbf{P}}_Y^\top\mathbf{M}_{*,z})}.
	\end{align*}
	Fix a pair of indices $x$ and $z$ and focus on the term
	\begin{align}\label{eq:term}
		\begin{split}
			&\quad			h(\mathbf{M}_{*,z})\\
			&=(\mathbf{P}_{x,*}^\top\mathbf{M}_{*,z})\log \frac{\hat{\mathbf{P}}_{x,*}^\top\mathbf{M}_{*,z}}{ (\mathbf{P}_{x,*}^\top\mathbf{1}_{|\mathcal{Y}|})(\hat{\mathbf{P}}_Y^\top\mathbf{M}_{*,z})}\\
			&=\underbrace{(\mathbf{P}_{x,*}^\top\mathbf{M}_{*,z})\log \frac{\hat{\mathbf{P}}_{x,*}^\top\mathbf{M}_{*,z}}{ \hat{\mathbf{P}}_Y^\top\mathbf{M}_{*,z}}}_{\text{first term } h_1(\mathbf{M}_{*,z})}-\underbrace{(\mathbf{P}_{x,*}^\top\mathbf{M}_{*,z})\log \mathbf{P}_{x,*}^\top\mathbf{1}_{|\mathcal{Y}|}}_{\text{second term}}.
		\end{split}
	\end{align}
	Note that the second term in~\eqref{eq:term} is linear in $\mathbf{M}_{*,z}$. Thus, if the first term $h_1(\mathbf{M}_{*,z})$ of $h(\mathbf{M}_{*,z})$ in~\eqref{eq:term} is convex in the variable $\mathbf{M}_{*,z}$, then the objective function $\hat{L}$ is a sum of convex functions and thus is itself convex. 
	
	For ease of exposition, let
	\begin{align*}
		&\mathbf{t}=\mathbf{M}_{*,z}, &&\mathbf{a}=\mathbf{P}_{x,*},\\
		&\mathbf{b}=\hat{\mathbf{P}}_{x,*},&&\mathbf{c}= \hat{\mathbf{P}}_Y,
	\end{align*}
	then, we have
	\begin{equation*}
		h_1(\mathbf{t})=(\mathbf{a}^\top\mathbf{t})\cdot\log\frac{\mathbf{b}^\top\mathbf{t}}{\mathbf{c}^\top\mathbf{t}},
	\end{equation*}
	where $\mathbf{a}^\top\mathbf{t}$, $\mathbf{b}^\top\mathbf{t}$ and $\mathbf{c}^\top\mathbf{t}$ are positive.
	Without loss of generality, we assume that $\mathbf{b}$ and $\mathbf{c}$ are linearly independent. Otherwise, $h_1(\mathbf{t})$ becomes a linear function of $\mathbf{t}$ and is convex. The derivative and Hessian of $h_1(\mathbf{t})$ are given by
	\begin{equation*}
		\nabla h_1(\mathbf{t})=\mathbf{a}\cdot\log(\frac{\mathbf{b}^\top\mathbf{t}}{\mathbf{c}^\top\mathbf{t}})+\mathbf{b}\cdot\frac{\mathbf{a}^\top\mathbf{t}}{\mathbf{b}^\top\mathbf{t}}-\mathbf{c}\cdot\frac{\mathbf{a}^\top\mathbf{t}}{\mathbf{c}^\top\mathbf{t}},
	\end{equation*}
	and
	\begin{align*}
		&\quad	\nabla^2 h_1(\mathbf{t})\\
		&=\frac{\mathbf{a}\mathbf{b}^\top}{\mathbf{b}^\top\mathbf{t}}-\frac{\mathbf{a}\mathbf{c}^\top}{\mathbf{c}^\top\mathbf{t}}+	
		\frac{\mathbf{b}\mathbf{a}^\top}{\mathbf{b}^\top\mathbf{t}}-\frac{\mathbf{a}^\top\mathbf{t}}{(\mathbf{b}^\top\mathbf{t})^2}\mathbf{b}\mathbf{b}^\top
		\\&\quad-	\frac{\mathbf{c}\mathbf{a}^\top}{\mathbf{c}^\top\mathbf{t}}+\frac{\mathbf{a}^\top\mathbf{t}}{(\mathbf{c}^\top\mathbf{t})^2}\mathbf{c}\mathbf{c}^\top\\
		&=\frac{\mathbf{a}\cdot((\mathbf{c}^\top\mathbf{t})\mathbf{b}-(\mathbf{b}^\top\mathbf{t})\mathbf{c})^\top+((\mathbf{c}^\top\mathbf{t})\mathbf{b}-(\mathbf{b}^\top\mathbf{t})\mathbf{c})\cdot\mathbf{a}^\top}{(\mathbf{b}^\top\mathbf{t})(\mathbf{c}^\top\mathbf{t})}\\
		&\quad+\frac{(\mathbf{a}^\top\mathbf{t})\cdot((\mathbf{b}^\top\mathbf{t})^2\mathbf{c}\mathbf{c}^\top-(\mathbf{c}^\top\mathbf{t})^2\mathbf{b}\mathbf{b}^\top)}{(\mathbf{b}^\top\mathbf{t})^2(\mathbf{c}^\top\mathbf{t})^2},\\
		&=\frac{\mathbf{a}\mathbf{u}^\top+\mathbf{u}\mathbf{a}^\top}{(\mathbf{b}^\top\mathbf{t})(\mathbf{c}^\top\mathbf{t})}-\frac{(\mathbf{a}^\top\mathbf{t})\cdot(\mathbf{u}\cdot((\mathbf{c}^\top\mathbf{t})\mathbf{b}+(\mathbf{b}^\top\mathbf{t})\mathbf{c})^\top)}{(\mathbf{b}^\top\mathbf{t})^2(\mathbf{c}^\top\mathbf{t})^2}\\
		&\quad+\frac{(\mathbf{a}^\top\mathbf{t})\cdot(\mathbf{b}\mathbf{c}^\top-\mathbf{c}\mathbf{b}^\top)}{(\mathbf{b}^\top\mathbf{t})(\mathbf{c}^\top\mathbf{t})},
	\end{align*}
	where $\mathbf{u}=(\mathbf{c}^\top\mathbf{t})\mathbf{b}-(\mathbf{b}^\top\mathbf{t})\mathbf{c}$. In order to show that $	\nabla^2h_1(\mathbf{t})$ is positive semidefinite under condition~\eqref{eq:convexitycondition}, we show that for any $ \mathbf{v}\in\mathbb{R}^{|\mathcal{Y}|}$, we have that $\mathbf{v}^\top\nabla^2h_1(\mathbf{t})\mathbf{v}\geq0$. We consider two cases.
	\begin{enumerate}
		\item $\mathbf{v}\perp \mathbf{u}$, in this case we have
		\begin{align}
			\mathbf{v}^\top\nabla_{\mathbf{t}}^2f(\mathbf{t})\mathbf{v}=0.
		\end{align}
		\item $\mathbf{v}=\mathbf{u}$, in this case we have
		\begin{align*}
			&\quad(\mathbf{b}^\top\mathbf{t})(\mathbf{c}^\top\mathbf{t})^2\frac{\mathbf{v}^\top\nabla^2h_1(\mathbf{t})\mathbf{v}}{\mathbf{v}^\top\mathbf{v}}\\
			&=2(\mathbf{c}^\top\mathbf{t})\cdot\mathbf{a}^\top\mathbf{v}-\frac{\mathbf{a}^\top\mathbf{t}}{\mathbf{b}^\top\mathbf{t}}((\mathbf{c}^\top\mathbf{t})\mathbf{b}+(\mathbf{b}^\top\mathbf{t})\mathbf{c})^\top\mathbf{v}\\
			&=2(\mathbf{c}^\top\mathbf{t})^2\mathbf{a}^\top\mathbf{b}-
			2(\mathbf{c}^\top\mathbf{t})(\mathbf{b}^\top\mathbf{t})\mathbf{a}^\top\mathbf{c}
			\\
			&\quad-\frac{\mathbf{a}^\top\mathbf{t}}{\mathbf{b}^\top\mathbf{t}}(\mathbf{c}^\top\mathbf{t})^2\mathbf{b}^\top\mathbf{b}+			
			\frac{\mathbf{a}^\top\mathbf{t}}{\mathbf{b}^\top\mathbf{t}}(\mathbf{b}^\top\mathbf{t})^2\mathbf{c}^\top\mathbf{c}\\
			&=(2\mathbf{a}^\top\mathbf{b}-\frac{\mathbf{a}^\top\mathbf{t}}{\mathbf{b}^\top\mathbf{t}}\mathbf{b}^\top\mathbf{b})(\mathbf{c}^\top\mathbf{t})^2-
			2(\mathbf{c}^\top\mathbf{t})(\mathbf{b}^\top\mathbf{t})\mathbf{a}^\top\mathbf{c}
			\\
			&\quad+			
			\frac{\mathbf{a}^\top\mathbf{t}}{\mathbf{b}^\top\mathbf{t}}\mathbf{c}^\top\mathbf{c}(\mathbf{b}^\top\mathbf{t})^2\\
			&=\begin{bmatrix}
				\mathbf{a}^\top\mathbf{t}&\mathbf{b}^\top\mathbf{t}
			\end{bmatrix}\begin{bmatrix}
				2\mathbf{a}^\top\mathbf{b}-\frac{\mathbf{a}^\top\mathbf{t}}{\mathbf{b}^\top\mathbf{t}}\mathbf{b}^\top\mathbf{b}&-\mathbf{a}^\top\mathbf{c}\\
				-\mathbf{a}^\top\mathbf{c}&\frac{\mathbf{a}^\top\mathbf{t}}{\mathbf{b}^\top\mathbf{t}}\mathbf{c}^\top\mathbf{c}
			\end{bmatrix}\begin{bmatrix}
				\mathbf{a}^\top\mathbf{t}\\\mathbf{b}^\top\mathbf{t}
			\end{bmatrix},
		\end{align*}
		which is nonnegative if
		\begin{equation}\label{eq:proofconvexity}
			(2\mathbf{a}^\top\mathbf{b}-\frac{\mathbf{a}^\top\mathbf{t}}{\mathbf{b}^\top\mathbf{t}}\mathbf{b}^\top\mathbf{b})(\frac{\mathbf{a}^\top\mathbf{t}}{\mathbf{b}^\top\mathbf{t}}\mathbf{c}^\top\mathbf{c})\geq(\mathbf{a}^\top\mathbf{c})^2,\quad\forall \mathbf{t}\in\mathbb{R}^{|\mathcal{Y}|}_{\geq0}.
		\end{equation}
		Note that the left hand side of~\eqref{eq:proofconvexity} is a quadratic equation in $\frac{\mathbf{a}^\top\mathbf{t}}{\mathbf{b}^\top\mathbf{t}}$, which, by Lemma~\ref{lemma:maxratio},  takes value in $[r_x^{\textup{L}},r_x^{\textup{U}}]$. Therefore, if~\eqref{eq:convexitycondition} is satisfied, then~\eqref{eq:proofconvexity} holds at the two end points of $[r_x^{\textup{L}},r_x^{\textup{U}}]$, which implies that~\eqref{eq:proofconvexity} holds for any $\mathbf{t}\in\mathbb{R}^{|\mathcal{Y}|}_{\geq0}$.
	\end{enumerate}
	In summary, if~\eqref{eq:convexitycondition} is satisfied, then $h_1(\cdot)$ is convex and so are $h(\cdot)$ and $\hat{L}$, which completes the proof.	
\end{proof}

\section{A useful lemma}
\begin{lema}[Bounds for a ratio]\label{lemma:maxratio}
	Given $\mathbf{u}\in\mathbb{R}^n_{\geq0}$ and $\mathbf{v}\in\mathbb{R}^n_{\geq0}$ that satisfy
	\begin{enumerate}[label=(\alph*)]
	\item\label{itm:cond1} $\mathbf{u}\neq \mathbf{0}$ and $\mathbf{v}\neq\mathbf{0}$;
	\item\label{itm:cond2} for any $i\in\{1,\dots,n\}$, if $\mathbf{u}_i>0$, then $\mathbf{v}_i>0$.
	\end{enumerate}
Let 
	\begin{equation*}
	r^{\textup{L}}=\min_{i:\mathbf{u}_i>0}\frac{\mathbf{u}_i}{\mathbf{v}_i}\quad\textup{and}\quad r^{\textup{U}}=\max_{i:\mathbf{u}_i>0}\frac{\mathbf{u}_i}{\mathbf{v}_i}.
\end{equation*}
	Then for any $\mathbf{w}\in\mathbb{R}^n_{\geq0}$ such that $\mathbf{w}^\top\mathbf{u}>0$, we have
	\begin{equation*}
		r^{\textup{L}}\leq\frac{\mathbf{w}^\top\mathbf{u}}{\mathbf{w}^\top\mathbf{v}}\leq r^{\textup{U}}.
	\end{equation*}
Moreover, for any $r\in[r^{\textup{L}},r^{\textup{U}}]$, there exists a $\mathbf{w}\in\mathbb{R}^n_{\geq0}$ with $\mathbf{w}^\top\mathbf{u}>0$ such that $\frac{\mathbf{w}^\top\mathbf{u}}{\mathbf{w}^\top\mathbf{v}}=r$.
\end{lema}
\begin{proof}	
	Let $i^*$ be such that $i^*=\argmax_{i:\mathbf{u}_i>0}\frac{\mathbf{u}_i}{\mathbf{v}_i}$. By conditions~\ref{itm:cond1} and \ref{itm:cond2}, such an $i^*$ always exists. Then, for any $j\in\{1,\dots,n\}$ such that $\mathbf{u}_j>0$ (and thus $\mathbf{v}_j>0$ by~\ref{itm:cond2}), we have that
	\begin{align}\label{eq:inequv}
		\frac{\mathbf{u}_{i^*}}{\mathbf{v}_{i^*}}\geq\frac{\mathbf{u}_{j}}{\mathbf{v}_{j}}&\implies\mathbf{u}_{i^*}\mathbf{v}_{j}\geq\mathbf{v}_{i^*}\mathbf{u}_{j}\nonumber\\
		&\implies\mathbf{u}_{i^*}\mathbf{v}_{j}\mathbf{w}_j\geq\mathbf{v}_{i^*}\mathbf{u}_{j}\mathbf{w}_j.
	\end{align}
	Summing over $j$ on both sides of~\eqref{eq:inequv}, we have
	\begin{equation*}		
		\mathbf{u}_{i^*}\sum_{j}\mathbf{v}_{j}\mathbf{w}_j\geq\mathbf{v}_{i^*}\sum_{j}\mathbf{u}_{j}\mathbf{w}_j
		\implies
		\frac{\mathbf{w}^\top\mathbf{u}}{\mathbf{w}^\top\mathbf{v}}\leq\frac{\mathbf{u}_{i^*}}{\mathbf{v}_{i^*}}=r^{\textup{U}}.
	\end{equation*}
Similarly, we can show that $r^{\textup{L}}\leq\frac{\mathbf{w}^\top\mathbf{u}}{\mathbf{w}^\top\mathbf{v}}$.

For any $r\in[r^{\textup{L}},r^{\textup{U}}]$, there exists a $\lambda\in[0,1]$, such that
\begin{align*}
	r&=\lambda r^{\textup{L}}+(1-\lambda)r^{\textup{U}}.
\end{align*}
Let $i^*$ and $j^*$ be such that $\frac{\mathbf{u}_{i^*}}{\mathbf{v}_{i^*}}=r^{\textup{U}}$ and $\frac{\mathbf{u}_{j^*}}{\mathbf{v}_{j^*}}=r^{\textup{L}}$. 
\begin{enumerate}
\item When $\lambda = 0$ ($\lambda=1$), we pick $\mathbf{w}_{i^*}=1$ ($\mathbf{w}_{j^*}=1$) and $\mathbf{w}_{j}=0$ for all $j\neq i^*$ ($\mathbf{w}_{j}=0$ for all $j\neq j^*$).
\item When $\lambda\in(0,1)$, let $\mathbf{w}_j=0$ for $j\notin\{i^*,j^*\}$,  $\mathbf{w}_{j^*}=1$ and 
\begin{equation*}
\mathbf{w}_{i^*}=\frac{\lambda}{1-\lambda}\frac{\mathbf{v}_{j^*}}{\mathbf{v}_{i^*}}.
\end{equation*}
Then,
\begin{align*}
\frac{\mathbf{w}^\top\mathbf{u}}{\mathbf{w}^\top\mathbf{v}}&=\frac{\frac{\lambda}{1-\lambda}\frac{\mathbf{v}_j^*}{\mathbf{v}_i^*}\mathbf{u}_{i^*}+\mathbf{u}_{j^*}}{\frac{\lambda}{1-\lambda}\frac{\mathbf{v}_j^*}{\mathbf{v}_i^*}\mathbf{v}_{i^*}+\mathbf{v}_{j^*}}\\
&=\frac{\frac{\lambda}{1-\lambda}\frac{\mathbf{v}_j^*}{\mathbf{v}_i^*}\mathbf{u}_{i^*}+\mathbf{u}_{j^*}}{\frac{1}{1-\lambda}\mathbf{v}_{j^*}}\\
&=\lambda\frac{\mathbf{u}_{i^*}}{\mathbf{v}_{i^*}}+(1-\lambda)\frac{\mathbf{u}_{j^*}}{\mathbf{v}_{i^*}}=r.
\end{align*}
\end{enumerate}
The proof is completed.
\end{proof}

\bibliographystyle{plainurl}
\bibliography{alias,mybib}

\end{document}